\documentclass[journal]{IEEEtran}
\usepackage{amssymb}
\usepackage[cmex10]{amsmath}
\usepackage{stfloats}
\usepackage{graphicx}
\usepackage{subfigure}
\usepackage{tabularx}
\usepackage{epsfig,epsf,color,balance,cite}
\usepackage{verbatim}
\usepackage{url}
\usepackage{bm}
\usepackage{amsthm}

\newtheorem{theorem}{Theorem}

\newtheorem{lemma}{Lemma}
\usepackage{algorithm}
\usepackage{algorithmic}

\usepackage{color}
\definecolor{myc1}{rgb}{0,0,0}
\definecolor{myc2}{rgb}{0,0,1}

\begin{document}

\title{Energy Efficient UAV Communication with Energy Harvesting}

\author{
\IEEEauthorblockN{Zhaohui Yang,
                  Wei Xu,
                    and
                  Mohammad Shikh-Bahaei
                  }
\thanks{Copyright (c) 2015 IEEE. Personal use of this material is permitted. However, permission to use this material for any other purposes must be obtained from the IEEE by sending a request to pubs-permissions@ieee.org.}
\thanks{This work was supported in part by the Engineering and Physical Science Research Council (EPSRC) through the Scalable Full Duplex Dense Wireless Networks (SENSE) grant EP/P003486/1,
 in part by the Natural Science Foundation of Jiangsu Province for Distinguished Young Scholars under Grant BK20190012, and in part by the NSFC under grants 61871109. (\emph{Corresponding author: Wei Xu}.)}
\thanks{Z. Yang and M. Shikh-Bahaei are with Centre for Telecommunications Research, King's College London, London WC2R 2LS, U.K. (Emails: \{yang.zhaohui, m.sbahaei\}@kcl.ac.uk).}
\thanks{ W. Xu is with the National Mobile Communications Research
Laboratory, Southeast University, Nanjing 210096, China  (Email: wxu@seu.edu.cn).}
}

\maketitle

\begin{abstract}
This paper investigates an unmanned aerial vehicle (UAV)-enabled wireless communication system with energy harvesting, where the UAV transfers energy to the users in half duplex or full duplex, and the users harvest energy for data transmission to the UAV. We minimize the total energy consumption of the UAV while accomplishing the minimal data transmission requests of the users. The original optimization problem is decomposed into two  subproblems: path planning subproblem and energy minimization subproblem with fixed path planning. For path planning subproblem, the optimal visiting order is obtained by using the dual method and the trajectory is optimized via the successive convex approximation technique. For energy minimization subproblem with fixed path planning, we firstly obtain the optimal portion of data transmission time within the entire procedure and the optimal transmission power of each user. Then, the the energy minimization subproblem is greatly simplified and it is efficiently solved via a one-dimensional search method. Simulation results are illustrated to verify the theoretical findings.
\end{abstract}

\IEEEpeerreviewmaketitle

\section{Introduction}

With the explosive growth of data traffic, unmanned aerial vehicle (UAV) communication has been deemed as a promising technology for future wireless communication networks \cite{7470933,DBLP:journals/corr/abs-1805-06532,saad2019vision}.
Both spectral and energy efficiency can be improved in scenarios where mobility of UAVs and  line-of-sight (LoS)-dominating ground channel characteristics are well explored \cite{7936620,8048502,8247211,8353131,8379427}.
More specifically,
UAVs can be utilized in various applications, such as data collection \cite{7469804,8119562,8432487,8316986}, wireless power transfer \cite{8489918,8417659,8540379,8736248,wu2019minimum,8365881,hu2018optimal,8434285}, relaying \cite{5937283,7959158,8278204,8629002}, device-to-device communications \cite{7412759}, caching \cite{8254370,7875131},  and mobile edge computing \cite{8764580}.

One line of work in the existing literature about UAV communication is UAV-aided ubiquitous coverage \cite{6863654,7918510,8038014,7762053,7510820,7486987}, where the UAVs are deployed to assist existing terrestrial communication infrastructure.
To fully exploit the degrees of freedom for designing UAV-enabled communications, it is crucial to investigate resource allocation in UAV-enabled wireless communication networks.
In \cite{6863654}, the altitude of UAV was optimized to provide maximum coverage on the ground.
To maximize the coverage using the minimum transmit power, an optimal location and altitude placement algorithm  was investigated in \cite{7918510} for UAV-base stations (BSs).
With different quality-of-service (QoS) requirements of users, authors in \cite{8038014} studied the three-dimension (3D) UAV-BS placement that maximizes the number of users in the coverage.
Exploiting the flexibility of UAV placement, the number of UAVs required for serving a certain area was considered in \cite{7762053}.
To further consider network delay, the optimal placement and distribution of cooperative UAVs was presented in \cite{7569080}.
In delay-constrained communication scenarios, \cite{8438896} investigated the fundamental throughput-delay tradeoff in UAV-enabled communications.
The authors in \cite{8779596} solved the mission completion time optimization  for multi-UAV-enabled data collection. The UAV trajectory was optimized in \cite{8419316} for parameter estimation in wireless sensor networks.


On the other hand, energy saving is critical for UAV communications especially in Internet of Things applications \cite{7469804}.
In order to prolong the lifetime of a sensor network, wireless energy consumption was minimized in \cite{8119562}.
Under a more practical energy consumption model of the UAV, it was pointed out that the propulsion energy is much larger than the communication-related energy \cite{7888557}.
Therefore, to minimize the dominating component of energy consumption, the authors in \cite{8432487} minimized the total flight time of a UAV while allowing sensors to successfully upload a certain amount of data.
Further considering the energy consumption of both user and UAV, the tradeoff between the propulsion energy and the wireless energy of the served user was investigated in~\cite{8316986}.
There are two major differences between this paper and \cite{8316986}.
One difference is that this paper investigates the total energy minimization for the rotary-wing UAV, while the fixed-wing UAV was adopted in \cite{8316986}.
The other difference is that this paper considers the general multiuser case, while only single user was investigated in \cite{8316986}.

Recently, energy harvesting \cite{6907966,7081080,6623062,DBLP:journals/corr/abs-1803-07123,8241822,8294215} has received a great deal of attention  in  prolonging the lifetime of low-power devices.
{\color{myc1}{Different from conventional wireless powered communication network (WPCN),
UAV-enabled WPCN can exploit the mobility of UAVs to further improve the system performance \cite{8489918,8417659,8540379,8736248,wu2019minimum}.
In \cite{8489918} and \cite{8417659}, the minimal uplink throughput among all users was maximized for UAV-enabled WPCN.
Considering weighted harvest-then-transmit protocol, the sum throughput of all users was maximized in \cite{8540379}.
To further consider the tradeoff between mission completion time and energy consumption, the energy-time region was obtained via jointly optimizing the UAV trajectory, user scheduling and time allocation \cite{8736248}.
For multi-UAV-enabled WPCN, the minimal throught maximization problem was investigated in \cite{wu2019minimum}.}}
The total amount of harvested energy for all devices was maximized  during a finite charging period for UAV communication in \cite{8365881}, and alternatively the minimal harvested energy among all devices was optimized \cite{hu2018optimal}.
However, the contributions  in \cite{8489918,8417659,8540379,8736248,wu2019minimum,8365881,hu2018optimal} ignored the UAV height optimization.

In this paper, we study a rotary-wing UAV communication system with energy harvesting, where the propulsion energy is explored and UAV height is optimized. The UAV serves as a data collector for multiple users. It broadcasts wireless energy to each user, while the user utilizes the harvested energy to transmit data to the UAV.
The contributions of this paper are summarized as:

\begin{enumerate}
  \item We formulate the problem of energy  minimization that jointly optimizes the UAV trajectory, user transmission power, and mission completion time.
 For the communication between the UAV and each user, half-duplex (HD) and full-duplex (FD) modes are investigated.
 \item The analytical models for the propulsion energy consumption of  a rotary-wing UAV with acceleration and deceleration are derived for straight flight and vertical flight. To obtain the optimal visiting order of gathering data from all users, the dual method is adopted.
  \item  For HD, the optimal relationship between the energy harvesting time and the transmission time is revealed. For FD, the optimal transmission time is obtained.
These findings ensure that the optimal solution of UAV height to the energy minimization problem can be effectively obtained via a one-dimensional (1D) search.
\end{enumerate}

The rest of this paper is organized as follows.
In Section~$\text{\uppercase\expandafter{\romannumeral2}}$, we introduce the system model and problem formulation.
Path planning and energy minimization with fixed path planning
are addressed in Section $\text{\uppercase\expandafter{\romannumeral 3}}$ and Section $\text{\uppercase\expandafter{\romannumeral 4}}$, respectively.
Numerical results are shown in Section $\text{\uppercase\expandafter{\romannumeral5}}$
and conclusions are finally drawn in Section $\text{\uppercase\expandafter{\romannumeral6}}$.


\section{System Model and Problem Formulation}
{\color{myc1}{
We consider a rotary-wing UAV-enabled wireless communication system with one UAV serving a set $\mathcal K$ of $K$ users.
The UAV serves as a data collector gathering information data from all users.
In the downlink the UAV transfers wireless energy to charge the users, while in the uplink the users utilize the harvested energy to transmit wireless information to the UAV.
Without loss of generality, we consider a 3D Cartesian coordinate system such that
the location of user $k$ is fixed at $(x_k,y_k,0)$ and the initial location of the UAV, point A$_0$, is at $(x_0,y_0,H)$, as shown in Fig.~\ref{sys1fig0}.
The UAV  returns back to point A$_0$ after the entire procedure of energy transfer and uplink data reception for all $K$ users.

\begin{figure}
\centering
\includegraphics[width=3.5in]{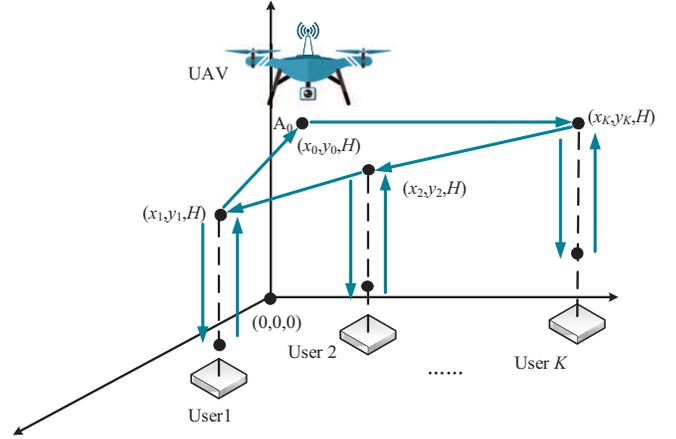}
\caption{A UAV-enabled wireless communication system.}\label{sys1fig0}
\end{figure}

To gather data from all users, the UAV visits all $K$ locations in $K+1$ stage.
In the first $K$ stages, the UAV sequentially collect data from  $K$ users.
The data collecting order is denoted by $\pi_1, \cdots, \pi_K$.
In the last $(K+1)$-th stage, the UAV flies back to the initial point and then it can start over again from the first user.}}


\subsection{Fly-Hover-Communicate Protocol}

\begin{figure}
\centering
\includegraphics[width=3.5in]{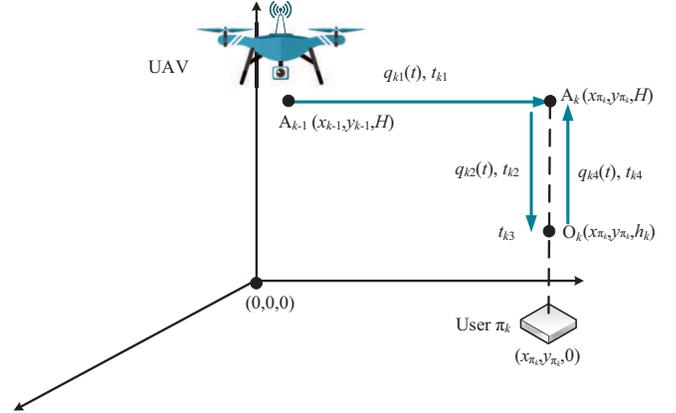}
\caption{Fly-hover-communicate protocol.}\label{sys1fig1}
\end{figure}

{\color{myc1}{The fly-hover-communicate protocol is adopted, i.e., the UAV first flies to a position close to the user and then the UAV broadcasts energy while the user transmits data.
In particular, the fly-hover-communicate protocol contains four steps when the UAV collects data from each user.

For the $k$-th $(k\leq K)$ stage, the UAV needs to collect data from user $\pi_k$.
In the first step, the UAV straight flies\footnote{\color{myc1}{The straight flight is considered because straight line is the shortest between two points. Thus, the straight flight is always energy saving for propulsion energy consumption.}} from the location A$_{k-1}$ $(x_{\pi_{k-1}},y_{\pi_{k-1}},H)$ to a  point A$_k$ $(x_{\pi_{k}},y_{\pi_{k}},H)$ with trajectory $q_{k1}(t)$ and time duration $t_{k1}$, as shown in Fig.~\ref{sys1fig1}.}}
In the second step,  to increase the energy harvesting efficiency, the UAV decreases its height to point O$_k$ $(x_{\pi_k}, y_{\pi_k}, h_k)$ with trajectory $q_{k2}(t)$ and time duration $t_{k2}$.
In the third step, the UAV staying stably at point O$_k$ communicates with the user with time duration $t_{k3}$.
In the fourth step, the UAV increases its hight from point O$_k$ to point A$_k$ at height $H$  with trajectory $q_{k4}(t)$ and  time duration $t_{k4}$.
Point A$_k$ is viewed as the initial point to gather data from user $\pi_{k+1}$.
Note that the UAV increases its height to $H$ because it is the allowed flying height to avoid collision with other UAVs or high buildings.
Since the trajectories $q_{k1}(t)$, $q_{k2}(t)$, and $q_{k4}(t)$ are one dimension, we can set the initial location as $q_{k1}(0)=q_{k2}(0)=q_{k4}(0)=0$ by choosing different reference points.

According to Appendix A, propulsion energy for straight flight in the first step can be expressed as $E_1(q_{k1}(t), t_{k1})$ defined in \eqref{apArweq12}.
For boundary constraints in $q_{k1}(t)$, we have
\begin{equation}
q_{k1}(t_{k1})=d_{\pi_{k-1}\pi_k},
v_{k1}(0)=v_{k1}(t_{k1})=0,
\end{equation}
where $d_{\pi_{k-1}\pi_k}=\sqrt{(x_{\pi_k}-x_{\pi_{k-1}})^2+(y_{\pi_k}-y_{\pi_{k-1}})^2}$ is the distance between A$_{k-1}$ and A$_{k}$, $v_{k1}(t)$ is the velocity at time $t$ and we denote $\pi_0=0$.
In the second step,  the vertical  descent energy   is $E_2(q_{k2}(t), t_{k2})$ defined in \eqref{apArweq23}.
For boundary constraints in $q_{k2}(t)$, we have
\begin{equation}
q_{k2}(t_{k2})=H-h_k, v_{k2}(0)=v_{k2}(t_{k2})=0.
\end{equation}
{\color{myc1}{ By substituting $q(t)=0$ and $T_0=t_{k3}$ into \eqref{apArweq12}, the hover energy consumption in the third step is $E_3(t_{k3})=(P_0+P_1)t_{k3}$.}}
In the fourth step, the vertical climb energy can be given by $E_4(q_{k4}(t), t_{k4})$ defined in \eqref{apArweq24}.
For boundary constraints in $q_{k4}(t)$, we also have
\begin{equation}
q_{k4}(t_{k4})=H-h_k, v_{k4}(0)=v_{k4}(t_{k4})=0.
\end{equation}


The total energy consumption of the UAV includes two components: the communication-related energy due to radiation, signal processing and other circuit, and the propulsion energy for ensuring the UAV to remain aloft.
In practice, the propulsion energy is much larger than the communication-related energy which can therefore be ignored \cite{7888557} in this paper.
Thus, the total energy consumption of the UAV in the $k$-th ($k\leq K$) stage  is
\begin{equation}\label{sys1sf1eq8_1}
E_k^{\text{Tot}} =\sum_{s\in\mathcal S}E_s(q_{ks}(t), t_{ks})+(P_0+P_1)t_{k3},
\end{equation}
where $\mathcal S=\{1,2,4\}$.

For the $(K+1)$-th stage, the UAV flies  from A$_K$ back to the initial point A$_0$ with trajectory $q_{(K+1)1}(t)$ and time duration $t_{(K+1)1}$.
As a result, the total energy consumption of the UAV in the $(K+1)$-th stage amounts to
\begin{align}\label{sys1sf1eq8_2}
E_{K+1}^{\text{Tot}}  = E_1(q_{(K+1)1}(t), t_{(K+1)1}).
\end{align}
For boundary constraints in $q_{K+1}(t)$, we have
\begin{equation}
q_{(K+1)1}(t_{K+1})=d_{ \pi_{K}\pi_{K+1}},
v_{K+1}(0)=v_{K+1}(t_{(K+1)1})=0,
\end{equation}
where we denote $\pi_{K+1}=0$.

For the third step in the $k$-th stage $(k\leq K)$ with staying stably at point O$_k$, the UAV broadcasts energy and receives information in HD or FD mode.

\subsubsection{HD Mode}
In HD mode, the UAV broadcasts energy to user $k$ with power $P$ in time duration $\rho_kt_{k3}$, and then user $k$ utilizes the harvested energy to upload a fixed amount of data to the UAV in time duration $(1-\rho_k)t_{k3}$, where $\rho_k\in[0,1]$ is the time splitting factor for user $\pi_k$.

{\color{myc1}{For UAV-ground links, large-scale attenuation is usually modelled as a random
variable depending on the occurrence probabilities of LoS and non-LoS.
The UAV is equipped with directional antenna.
The channel gain between the UAV at point O$_k$ and user $\pi_k$ is denoted by $g_{k}$, which can be expressed as \cite{7412759,6863654,zeng2019energy}
\begin{equation}\label{sys1sf1eqgk}
g_k=
\left\{ \begin{array}{ll}
\!\!\!   \sqrt{\beta_0 h_k^{-\alpha} } z_k,
  &   \textrm{with probability } a, \\
\!\!\!   \sqrt{\kappa\beta_0 h_k^{-\alpha} } z_k,
 &  \textrm{with probability } 1-a.\\
\end{array} \right.
\end{equation}
In \eqref{sys1sf1eqgk}, $\beta_0$ is  channel power gain at the reference distance 1~m, $\alpha$ is the path
loss exponent,  $\kappa$ is the additional attenuation factor
due to the non-LoS condition,  the small-scale fading
$z_k$ is  exponentially distributed with parameter 1,
and $a$ is the LoS probability between the UAV and user $\pi_k$.
In particular, $a$ can be given by \cite{7412759}
\begin{equation}\label{sys1sf1eqgk2}
a=\frac{1}{1+C_1(\exp(-C_2(90-C_1)))},
\end{equation}
where $C_1$ and $C_2$ are parameters depending on the propagation
environment, and $90$ is the elevation angle.}}
The expected channel power gain by averaging over
both randomness of small scaling fading and LoS occurrence is
\begin{equation}\label{sys1sf1eqgk2_2}
\mathbb E(|g_k|^2)=(a+\kappa(1-a)) \beta_0 h_k^{-\alpha}=b  h_k^{-\alpha},
\end{equation}
where we set $b=(a+\kappa(1-a))\beta_0$.

During the time duration of $\rho_k t_{k3}$, the harvested energy at user $\pi_k$ can be evaluated by \cite{6907966},
\begin{equation}\label{sys1sf1eq3}
e_{\pi_k}^{\text {HD}}=\zeta P \rho_k t_{k3} \mathbb E(|g_k|^2)=\zeta  P b  h_k^{-\alpha} \rho_k t_{k3},
\end{equation}
where $0<\zeta<1$ is the energy harvesting efficiency at user $k$.

During the time duration of $(1-\rho_k)t_{k3}$ for uplink data transfer, the average transmit power of user $\pi_k$   is $p_{\pi_k}$.
Then,
the transmitted data from  user $\pi_k$ to the UAV   can be accordingly given by
\begin{align}\label{sys1sf1eq6_2}
d_{\pi_k}^{\text {HD}} &= (1-\rho_k)t_{k3} B\log_2
\left( 1 + \frac {p_{k} |g_k|^2}{B \sigma^2}
\right),
\end{align}
where $B$ is the bandwidth of the system and $\sigma^2$ is the power spectral density of the additive white Gaussian noise.
Due to the randomness of $g_k$, $d_{k}^{\text {HD}}$ is a random variable.
According to the Jensen's inequality, we have
\begin{align}\label{sys1sf1eq6_3}
\mathbb E(d_{\pi_k}^{\text {HD}}) &\leq (1-\rho_k)t_{k3} B\log_2
\left( 1 + \frac {p_{\pi_k}  \mathbb E(|g_k|^2)}{B \sigma^2}
\right)
\nonumber\\
&= (1-\rho_k)
t_{k3} B\log_2
\left( 1 + \frac {p_{\pi_k} b  h_k^{-\alpha}}{B \sigma^2}
\right)\triangleq r_{\pi_k}^{\text{HD}}.
\end{align}
In the following, we use $r_{\pi_k}^{\text{HD}}\geq D_{\pi_k} $ to approximately represent the target data collection constraint, where $D_{\pi_k}$ is the minimal required data of user $\pi_k$.
Numerical results in \cite{zeng2019energy} show rather satisfactory accuracy for such approximation.

{\color{myc1}{Due to  energy harvesting in time duration $\rho_k t_{k3}$ and data transmission in duration $(1-\rho_k)t_{k3}$, the consumed energy at user $\pi_k$ is
\begin{equation}\label{sys1sf1eq3_1}
c_{\pi_k}^{\text {HD}}=\rho_k t_{k3} p_{\pi_k}^{\text{re}}+(1-\rho_k)t_{k3} \left(p_{\pi_k}^{\text{tr}}+ \frac{p_{\pi_k} }{\epsilon_{\pi_k}}\right),
\end{equation}
where $p_{\pi_k}^{\text{re}}$ is the receive circuit power consumption, $p_{\pi_k}^{\text{tr}}$ is the circuit power consumed for transmit signal processing, and $\epsilon_{\pi_k}\in(0,1]$ is a constant which accounts for the power
amplifier efficiency.}}

\subsubsection{FD Mode}
{\color{myc1}{In FD mode, the UAV broadcasts energy to user $\pi_k$ with power $P$, while user $\pi_k$ simultaneously uploads the data to the UAV in time duration $t_{k3}$ by using the harvested energy.
The transmitted data from user $\pi_k$ to the UAV in FD mode is
\begin{align}\label{sys1sf1eq8}
\mathbb E{(d_{\pi_k}^{\text{FD}})} &= \mathbb E\left((t_{k3}-\delta_{\pi_k}) B \log_2 \left(
1+\frac{p_{\pi_k}   |g_k|^2}{\gamma P +B\sigma^2}
\right)\right) \nonumber \\
&\leq  (t_{k3}-\delta_{\pi_k}) B \log_2 \left(
1+\frac{p_{\pi_k} b  h_k^{-\alpha}}{\gamma P +B\sigma^2}
\right)\triangleq r_{\pi_k}^{\text{FD}},
\end{align}
where $\gamma$ denotes the effective self-interference coefficient in FD operations and $\delta_{\pi_k}$ is the processing delay of the energy circuits of user $\pi_k$.
In time duration $t_{k3}$,
  the harvested energy at user $\pi_k$  can be evaluated by \cite{6907966},
\begin{equation}\label{sys1sf1eq3_2}
e_{\pi_k}^{\text {FD}}= \zeta  P b  h_k^{-\alpha}t_{k3},
\end{equation}
and the consumed energy at user $\pi_k$  is
\begin{equation}
c_{\pi_k}^{\text {FD}}=t_{k3}  p_{\pi_k}^{\text{re}}+ (t_{k3}-\delta_{\pi_k}) \left(p_{\pi_k}^{\text{tr}}+ \frac{p_{\pi_k} }{\epsilon_{\pi_k}}\right).
\end{equation}
\vspace{-3em}
}}

\subsection{Problem Formulation}
Now, it is ready to formulate the following total energy minimization problem with minimal user upload data demand as:
 \begin{subequations}\label{sys1min1}
\begin{align}
\mathop{\min}_{ \pmb \pi, \pmb q, \pmb t, {\color{myc1}{\pmb h}}, \pmb \rho, \pmb p   }\!\!\!&\:
\sum_{k=1}^{K+1} E_k^{\text{Tot}}  \\
\textrm{s.t.}\quad\:
 &r_{\pi_k}^{\text {HD}} \geq D_{\pi_k},  \forall k \in\mathcal K \\
& e_{\pi_k}^{\text {HD}} \geq c_{\pi_k}^{\text {HD}},   \forall k \in\mathcal K\\
& 0\leq \rho_k\leq 1,   \forall k \in\mathcal K\\
& q_{ks}(0)=q_{(K+1)1}(0)=0,   \forall k \in\mathcal K,s\in\mathcal S \\
& q_{k1}(t_{k1})=d_{\pi_{k-1}\pi_k},  \forall k \in\mathcal K' \\
&q_{ks}(t_{ks})=H-h_k, \forall k \in\mathcal K, s\!=\!2,4 \\
& v_{k1}(0)=v_{k1}(t_{k1})=0,  \forall k \in\mathcal K'\\
&v_{ks}(0)=v_{ks}(t_{ks})=0,  \forall k \in\mathcal K, s\!=\!2,4\\
&  v_{ks}(t)\!=\!\dot{q}_{ks}(t), a_{ks}(t)\!=\!\ddot{q}_{ks}(t),\! \forall k \!\in\!\mathcal K, \!s\!=\!2,4\\
& v_{k1}(t)\!=\!\dot{q}_{k1}(t), a_{k1}(t)\!=\!\ddot{q}_{k1}(t), \forall k \!\in\!\mathcal K'\\
&|v_{ks}(t)|\!\leq\! V_{\max},|a_{ks}(t)|\!\leq\! A_{\max}, \!\forall k \!\in\!\mathcal K', s\!\in\!\mathcal S\\
&\pmb \pi \in \Pi
\end{align}
\end{subequations}
for HD,
and
 \begin{subequations}\label{sys1min1_2}
\begin{align}
\mathop{\min}_{ \pmb \pi, \pmb q, \pmb t, {\color{myc1}{\pmb h}}, \pmb p   }\quad
&\sum_{l=1}^{K+1} E_l^{\text{Tot}}  \\
\textrm{s.t.}\quad
& r_{\pi_k}^{\text {FD}} \geq D_{\pi_k},  \forall k \in\mathcal K \\
& e_{\pi_k}^{\text {FD}} \geq c_{\pi_k}^{\text {FD}},   \forall k \in\mathcal K\\
& (\ref{sys1min1}e)-(\ref{sys1min1}m),
\end{align}
\end{subequations}
for FD,
where $\pmb \pi=\{\pi_k\}$,
$\pmb q=\{q_{k1}(t),q_{k2}(t), q_{k4}(t)\}$,
$\pmb t=\{t_{k1}, t_{k2}, t_{k3}, t_{k4} \}$,
$\pmb h=\{h_k\}$,
$\pmb \rho=\{\rho_k\}$,
$\pmb p =\{p_k \}$,
$\mathcal K'=\mathcal K \cup\{K+1\}$,
$V_{\max}$ and $A_{\max}$ are respectively the maximal velocity and acceleration of the UAV,
 and $\Pi$ is the set of all possible permutation of all $K$ users.
Constraints (\ref{sys1min1}b) or (\ref{sys1min1_2}b) ensure successful data collection, while
constraints (\ref{sys1min1}c) or (\ref{sys1min1_2}c) mean that the consumed energy of each user should be less than the harvested energy.
The time splitting factor constraints are given in  (\ref{sys1min1}d).
Constraints (\ref{sys1min1}e)-(\ref{sys1min1}i) are the boundary constraints for the trajectory.

Note that the trajectory from $A_k$ to $O_k$ always exists since the data from all users should be collected.
Moreover,  the discrete data collecting order variable $\pmb \pi$ only affects the boundary constraints of fist-step trajectory in each fly-hover-communicate stage.
This motivates us to decouple the energy minimization problem with two subproblems without loss of generality, i.e., path planning subproblem with variable ($\pmb \pi, \bar{\pmb q}=\{q_{k1}(t)\},
\bar{\pmb t}=\{t_{k1}\}$) and energy minimization subproblem with variable ($ \tilde{\pmb q}=\{q_{k2}(t), q_{k4}(t)\},
\tilde{\pmb t}=\{t_{k2},t_{k3},t_{k4}\}, \pmb h, \pmb \rho,\pmb p $) for HD mode and ($ \tilde{\pmb q},
\tilde{\pmb t}, \pmb h, \pmb p $) for FD mode.

\section{Path Planning}
In this section, the path planning problem is investigated.
According to \eqref{sys1sf1eq8_1}, \eqref{sys1sf1eq8_2}, problems (\ref{sys1min1}) and (\ref{sys1min1_2}), the path planning subproblem can be formulated as
 \begin{subequations}\label{pp1min1}
\begin{align}
\mathop{\min}_{ \pmb \pi, \bar{\pmb q}, \bar{\pmb t} } \quad
&\sum_{k=1}^{K+1} E_1(q_{k1}(t), t_{k1})  \\
\textrm{s.t.}\quad
& q_{k1}(0)=0, q_{k1}(t_{k1})=d_{\pi_{k-1}\pi_k},  \forall k \in\mathcal K' \\
& v_{k1}(0)=v_{k1}(t_{k1})=0,  \forall k \in\mathcal K'\\
& v_{k1}(t)\!=\!\dot{q}_{k1}(t), a_{k1}(t)\!=\!\ddot{q}_{k1}(t), \forall k \!\in\!\mathcal K'\\
&|v_{k1}(t)|\leq V_{\max},|a_{k1}(t)|\leq A_{\max}, \forall k \!\in\!\mathcal K'\\
&\pmb \pi \in \Pi.
\end{align}
\end{subequations}
Based on (\ref{pp1min1}b), it is observed that the trajectory is a function of $d_{\pi_k\pi_{k-1}}$.
Due to this observation, problem \eqref{pp1min1} is further equivalent to
 \begin{subequations}\label{pp1min1_1}
\begin{align}
\mathop{\min}_{ \pmb \pi } \quad
&\sum_{k=1}^{K+1} E_{\pi_{k-1}\pi_{k}}  \\
\textrm{s.t.}\quad
&\pmb \pi \in \Pi,
\end{align}
\end{subequations}
where
\begin{subequations}\label{pp1min1_2}
\begin{align}
E_{\pi_{k}\pi_{k-1}}=\mathop{\min}_{  q_{k1} (t),t_{k1} } \:
&  E_1(q_{k1}(t), t_{k1})  \\
\textrm{s.t.}\quad
&q_{k1}(0)=0, q_{k1}(t_{k1})=d_{\pi_{k-1}\pi_k}    \\
& v_{k1}(0)=v_{k1}(t_{k1})=0  \\
& v_{k1}(t) =\dot {q}_{k1}(t), a_{k1}(t) =\ddot{q}_{k1}(t)\\
&|v_{k1}(t)|\!\leq\! V_{\max},|a_{k1}(t)|\!\leq\! A_{\max}.
\end{align}
\end{subequations}
In \eqref{pp1min1_2}, $E_{\pi_{k-1}\pi_{k}}$ means the minimal propulsion energy consumption with flying distance $d_{\pi_{k-1}\pi_k}$.

{\color{myc1}{
For problem \eqref{pp1min1_2}, in order to construct a feasible solution, we consider two special cases:
1) the UAV first increases the speed from zero to a specific speed less than $V_{\max}$ with the maximal acceleration $a_{\max}$ in time duration $t_0$, and then the UAV decreases the speed to zero with acceleration $-a_{\max}$;
2) the UAV first increases the speed from zero to $V_{\max}$ with acceleration $a_{\max}$ in time duration $t_1$, then the UAV flies with the constant speed, finally, the UAV decreases its speed to zero  with acceleration $-a_{\max}$, as shown in Fig.~\ref{sys1fig3}.
Constraints (\ref{pp1min1_2}c)-(\ref{pp1min1_2}e) can be easily satisfied if the UAV trajectory is determined.
To satisfy constraint (\ref{pp1min1_2}b), we consider the following two situations, separately.

1) If $d_{\pi_{k-1}\pi_k}<\frac{V_{\max}^2}{a_{\max}}$, the UAV trajectory follows case 1 in  Fig.~\ref{sys1fig3}.
In this case, we have $ {a_{\max}t_{0}^2} =d_{\pi_{k-1}\pi_k}$ and $t_{k1}=2t_0$, i.e.,
\begin{equation}\label{pp1min1_2feq1}
t_{k1}=2\sqrt{\frac{d_{\pi_{k-1}\pi_k}}{a_{\max}}},
\end{equation}
which implies that a feasible solution of (\ref{pp1min1_2}) always exists.

2) If $d_{\pi_{k-1}\pi_k}\geq\frac{V_{\max}^2}{a_{\max}}$, the UAV trajectory follows case 2 in  Fig.~\ref{sys1fig3}.
In this case, we have $t_1=\frac{V_{\max}}{a_{\max}}$, $\frac{V_{\max}^2}{a_{\max}} +(t_2-t_1)V_{\max} =d_{\pi_{k-1}\pi_k}$ and $t_{k1}=t_1+t_2$, i.e.,
\begin{equation}\label{pp1min1_2feq2}
t_{k1}=  {\frac{d_{\pi_{k-1}\pi_k}}{V_{\max}}} + {\frac{V_{\max}}{a_{\max}}},
\end{equation}
which ensures a feasible solution of (\ref{pp1min1_2}).

\begin{figure}
\centering
\includegraphics[width=3.5in]{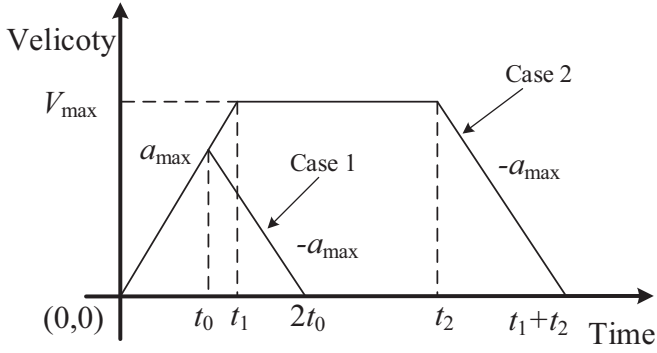}
\caption{The velicoty versus time.}\label{sys1fig3}
\end{figure}

\subsection{Solution for Problem \eqref{pp1min1_1}}
To rewrite problem \eqref{pp1min1_1} in a simplifier manner, we introduce new variable $w_{kl}$ to represent that the UAV collects data from user $l$ in the $k$-th stage.
With new variable $\pmb w=\{w_{kl}\}$, problem \eqref{pp1min1_1}  can be equivalently transformed to
\begin{subequations}\label{pp1min2_1}
\begin{align}
\mathop{\min}_{ \pmb w}\quad
& \sum_{l=1}^K w_{1l} E_{0l} + \sum_{k=2}^{K} \sum_{l=1}^K  \sum_{i=1}^K w_{(k-1)l} w_{ki} E_{li}
\nonumber\\ & +\sum_{l=1}^K w_{Kl} E_{l0} \\
\textrm{s.t.}\quad
& \sum_{k=1}^K w_{kl}=1,   \forall l\in\mathcal  K\\
& \sum_{l=1}^K w_{kl}=1,   \forall k\in\mathcal K\\
& w_{kl}\in\{0,1\},  \forall k,l\in\mathcal K.
\end{align}
\end{subequations}
In the objective function (\ref{pp1min2_1}a),  $\sum_{l=1}^K w_{1l} E_{0l}$ means the energy consumption in the first stage.
For $2\leq k\leq K$, the term $w_{(k-1)l} w_{ki} E_{li}$ means the energy consumption in the $k$-th stage when the UAV visits user $l$ in the $(k-1)$-th stage and visits user $i$ in the $k$-th stage.
Summing all possible cases, $\sum_{l=1}^K  \sum_{i=1}^K w_{(k-1)l} w_{ki} E_{li} $ means the energy consumption in the $k$-th stage.
The last term in (\ref{pp1min2_1}a), $\sum_{l=1}^K w_{Kl} E_{l0}$ stands for the energy consumption in the last $(K+1)$-th stage since the UAV flies back to the initial point in this stage.
Constraints  (\ref{pp1min2_1}b) represent that  each user is visited at once, while constraints (\ref{pp1min2_1}c) show that only one user is visited at each stage.

Problem \eqref{pp1min2_1} is a non-linear integer problem due to non-linear term $w_{(k-1)l} w_{ki}$ in the objective function.
To transform problem \eqref{pp1min2_1} into an equivalent solvable form, we introduce new variable $v_{kli}=w_{(k-1)l} w_{ki}$.
Due to the fact that $w_{kl}\in\{0,1\}$,   constraint $v_{kli}=w_{(k-1)l} w_{ki}$ is equivalent to
\begin{equation}\label{pp1min2_1eq1}
v_{kli}\geq w_{(k-1)l}+ w_{ki}-1,v_{kli}\geq 0,
\end{equation}
and
\begin{equation}\label{pp1min2_1eq2}
v_{kli}\leq w_{(k-1)l} , v_{kli}\leq w_{ki},
\end{equation}
for all $k=\mathcal K\setminus\{1\}$, $l,i\in\mathcal K$.

With new variable $\pmb v=\{v_{kli}\}$, problem \eqref{pp1min2_1} is equivalent to the following integer problem
\begin{subequations}\label{pp1min2_12}
\begin{align}
\mathop{\min}_{ \pmb w, \pmb v}\quad
& \sum_{l=1}^K w_{1l} E_{0l} + \sum_{k=2}^{K} \sum_{l=1}^K  \sum_{i=1}^K v_{kli}^2 E_{li}
\nonumber\\ & +\sum_{l=1}^K w_{Kl} E_{l0} \\
\textrm{s.t.}\quad
&  (\ref{pp1min2_1}b)-(\ref{pp1min2_1}d), \eqref{pp1min2_1eq1},\eqref{pp1min2_1eq2}.
\end{align}
\end{subequations}
Note that we use $v_{kli}^2$ to replace $v_{kli}$ in the objective function (\ref{pp1min2_12}a) because $v_{kli}$ is always 0 or 1 according to \eqref{pp1min2_1eq1} and \eqref{pp1min2_1eq2}.
The advantage of transforming problem \eqref{pp1min2_1} into problem \eqref{pp1min2_12} is that problem \eqref{pp1min2_12} with relaxed constraints is convex.

Due to integer constraints (\ref{pp1min2_1}d), it is hard to solve   problem \eqref{pp1min2_12}.
By temporarily relaxing the integer constraints (\ref{pp1min2_1}d) with $w_{kl}\in[0,1]$, problem \eqref{pp1min2_12} is a convex problem.
For convex problem \eqref{pp1min2_12} with relaxed constraints, the optimal solution can be effectively obtained by using the dual method \cite{boyd2004convex,8088359}.
We show that it fortunately gets integer solutions, which preserves both optimality and feasibility of the original problem even though
the relaxation is deployed temporarily.

To obtain the optimal solution of problem (\ref{pp1min2_12}), we have the following theorem.

\begin{theorem}
For problem (\ref{pp1min2_12}), the optimal path planning $\pmb w$ and auxiliary vector $\pmb v$ can be respectively expressed as
\begin{equation}\label{PAuser2eq2_2}
w_{kl}^*=\left\{ \begin{array}{ll}
\!\!1, &\text{if}\; k =\arg\min_{i\in\mathcal K} C_{il}  \\
\!\!0, &\text{otherwise},
\end{array} \right.
\end{equation}
and
\begin{equation}\label{PAuser2eq2_6_2}
v_{kli}^*=\left[\frac{\gamma_{kli} -\lambda_{kli} -\mu_{kli}}{2E_{li}}\right]^+,
\end{equation}
where
\begin{equation}\label{PAuser2eq2_2_2}
 C_{kl}\!=\!\left\{ \begin{array}{l}
\!\!\!E_{0l}+\beta_1+\sum_{i=1}^K(\gamma_{2li}-\lambda_{2li}), \quad  \text{if}\; k =1\\
\!\!\!\beta_k\!+\!\sum_{i=1}^K(\gamma_{(k+1)li}\!+\!\gamma_{kil}\!-\!\lambda_{(k+1)li}\!-\!\mu_{kil}),  \\
 \qquad\qquad\qquad\qquad\qquad\text{if}\; 2\leq  k \leq K-1\\
\!\!\!E_{l0}+\sum_{i=1}^K(\gamma_{Kil}-\mu_{Kil}),  \quad\text{if}\; k =K.
\end{array} \right.
\end{equation}
$\{\beta_k\}, \{\gamma_{kli}\},\{\lambda_{kli}\}, \{\mu_{kli}\}$ are  Lagrange multipliers associated with corresponding constraints of problem (\ref{pp1min2_12}),
and $[x]^+=\max\{x,0\}$.
If there are multiple minimal points in $\arg\min_{i\in\mathcal K} C_{il}$, we will choose any one of them.
\end{theorem}

\begin{proof}
Please refer to Appendix B.
\end{proof}

The values of $\{\beta_k\}, \{\gamma_{kli}\},\{\lambda_{kli}\}, \{\mu_{kli}\}$ can be determined by the sub-gradient method \cite{bertsekas2009convex}.
The updating procedure is given by
\begin{eqnarray}
&&\!\!\!\!\!\!\!\!\!\!\!\!\!\!\!\!\!\!
\beta_k= \beta_k+ \phi   \left(\sum_{l=1}^K w_{kl}-1\right)\label{PAuser2eq2_3}\\
&&\!\!\!\!\!\!\!\!\!\!\!\!\!\!\!\!\!\!
\gamma_{kli} =\left[\gamma_{kli}  + \phi ( w_{(k-1)l}+ w_{ki}-1- v_{kli} )  \right]^+\\
&&\!\!\!\!\!\!\!\!\!\!\!\!\!\!\!\!\!\!
\lambda_{kli}=\left[\lambda_{kli}+ \phi(v_{kli}-w_{(k-1)l})\right]^+\\
&&\!\!\!\!\!\!\!\!\!\!\!\!\!\!\!\!\!\!
\mu_{kli} =\left[ \mu_{kli} + \phi(v_{kli}- w_{ki})\right]^+ ,\label{PAuser2eq2_5}
\end{eqnarray}
where $\phi>0$ is a dynamically chosen step-size sequence.

By iteratively optimizing primal variable and  dual variable, the optimal path planning is obtained.
The dual method used to obtain the optimal path planning is given in Algorithm 1.
Notice that the optimal $w_{kl}$ is either 0 or 1 according to \eqref{PAuser2eq2_2}, even though
we relax $w_{kl}$ as continuous constraint in (\ref{pp1min2_12}). Therefore, we
can obtain optimal solution to problem (\ref{pp1min2_12}).

 \begin{algorithm}[h]
\caption{  Dual Method for Problem (\ref{pp1min2_12}) }
\begin{algorithmic}[1]
\STATE Initialize dual variables $\{\beta_k\}, \{\gamma_{kli}\},\{\lambda_{kli}\}, \{\mu_{kli}\}$.
\REPEAT
\STATE Update path planning variable $\pmb w$ and auxiliary vector $\pmb v$  according to (\ref{PAuser2eq2_2})-(\ref{PAuser2eq2_2_2}).
\STATE Update dual variables based on (\ref{PAuser2eq2_3})-(\ref{PAuser2eq2_5}).
\UNTIL the objective value (\ref{pp1min2_12}a) converges
\end{algorithmic}
\end{algorithm}

For the complexity, according to Algorithm 1, the major complexity lies in updating $\pmb v$.
Based on \eqref{PAuser2eq2_6_2}, the complexity of calculating $\pmb v$ is
$\mathcal O(K^3)$.
As a result, the total complexity of Algorithm 1 is $\mathcal O(LK^3)$, where $L$ is the number of iterations in Algorithm 1.

\subsection{Solution for Problem \eqref{pp1min1_2}}

Problem \eqref{pp1min1_2} is difficult to be directly solved due to the following two difficulties.
The first difficulty is the unknown fly time $t_{k1}$.
The second difficulty is that problem \eqref{pp1min1_2} involves
 an infinite number of optimization variables in continuous-function $q_{k1}(t)$.

To handle the first difficulty, we provide the following lemma regarding the property about the optimal fly time.
Denote the optimal solution of problem \eqref{pp1min1_2} by $(q_{k1}^*(t),t_{k1}^*)$.
\begin{lemma}
With fixed fly time $t_{k1}>t_{k1}^*$, the optimal trajectory of problem \eqref{pp1min1_2} is
\begin{equation}\label{pp1min1_2eq1}
\bar q_{kl}(t)=\left\{ \begin{array}{ll}
\!\!q_{kl}^*(t), &\text{if}\; 0\leq t \leq  t_{k1}^* \\
\!\!q_{kl}^*( t_{k1}^*), &\text{if}\;  t_{k1}^* \leq t \leq t_{k1}.
\end{array} \right.
\end{equation}
\end{lemma}

\begin{proof}
According to \eqref{apArweq7}, it is observed that the propulsion power of the UAV with positive velocity is strictly larger than the propulsion power $P_0+P_1$ with zero velocity  (i.e., the UAV is stable).
Since $q_{k1}^*(t)$ is the optimal solution for problem \eqref{pp1min1_2} with fixed optimal fly $t_{k1}^*$, the optimal $\bar q_{kl}(t)$ of roblem \eqref{pp1min1_2} with $t_{k1}>t_{k1}^*$ should be given by \eqref{pp1min1_2eq1}.
\end{proof}

Based on Lemma 1, it is observed that the optimal solution of problem \eqref{pp1min1_2} can be constructed by solving problem \eqref{pp1min1_2} even with  fixed fly time ($t_{k1}>t_{k1}^*$).
For the upper bound of the optimal fly time, the following lemma is provided.

\begin{lemma}
Given a feasible solution of problem \eqref{pp1min1_2} with objective value $E$,  the optimal fly time $t_{k1}^*$ of problem \eqref{pp1min1_2} satisfies $t_{k1}^*<\frac{E}{P_0+P_1}$.
\end{lemma}

\begin{proof}
Denote $E^*$ as the optimal objective value of problem \eqref{pp1min1_2}.
Then, we have $E^* \leq E$.
Due to the fact that the propulsion power of the UAV with positive velocity is strictly larger than   $P_0+P_1$ according to \eqref{apArweq7},
we can obtain that $E^*> (P_0+P_1) t_{k1}^*$.
Thus, we have $t_{k1}^*<\frac{E}{P_0+P_1}$.
\end{proof}

Based on Lemma 1, problem \eqref{pp1min1_2} can be effectively solved without loss of optimality even with fixed fly time obtained from Lemma 2.

To handle the second difficulty, problem \eqref{pp1min1_2} can be reformulated with applying discrete linear state-space approximation.
By discretizing the time duration $t_{k1}$ into $N+1$ slots with step size $\delta=\frac{t_{k1}}{N}$, i.e., $t=n\delta$, $n=0,1,\cdots,N$, the UAV trajectory $q_{k1}(t)$ can be characterized by the discrete-time UAV location $\{q_{k1n}\}$.
For continuous constraints (\ref{pp1min1_2}d) with small step size $\delta$, we have the following equation about the velocity $\{v_{k1n}\}$ and acceleration $\{a_{k1n}\}$
\begin{equation}\label{pp1min1_2eq2}
v_{k1n}=\frac{q_{k1n}-q_{k1(n-1)}}{\delta}, \forall n \in\mathcal N
\end{equation}
\begin{equation}\label{pp1min1_2eq2_1}
a_{k1n}=\frac{v_{k1n}-v_{k1(n-1)}}{\delta}, \forall n \in \mathcal N,
\end{equation}
where $\mathcal N=\{1,\cdots,N\}$.
The maximal velocity and acceleration constraints (\ref{pp1min1_2}e) can be accordingly presented as
\begin{equation}\label{pp1min1_2eq3}
|v_{k1n}|\!\leq\! V_{\max},|a_{k1n}|\!\leq\! A_{\max}, \forall n\in\mathcal N.
\end{equation}

As a result, problem \eqref{pp1min1_2} can be expressed in the discrete form as
\begin{subequations}\label{pp1min1_2_1}
\begin{align}
\mathop{\min}_{  \{q_{k1n}\},\{v_{k1n}\},\{a_{k1n}\}}
&  \bar E_1(\{q_{k1n}\},\{v_{k1n}\},\{a_{k1n}\} )  \\
\textrm{s.t.}\quad
& q_{k10}=0,q_{k1N}=d_{\pi_{k-1}\pi_k}    \\
& v_{k10}=v_{k1N}=0  \\
&  (\ref{pp1min1_2eq2})-(\ref{pp1min1_2eq3}),
\end{align}
\end{subequations}
where
\begin{align}\label{pp1min1_2_1eq1}
 &\bar E_1(\{q_{k1n}\},\{v_{k1n}\},\{a_{k1n}\} )=\sum_{n=1}^N \Bigg[P_0 \left(1+ c_1 v_{k1n}^2 \right)  \nonumber\\& +  P_1  \sqrt{  1+\left(c_2  v_{k1n}^2+c_3 a_{k1n} v_{k1n}   \right)^2  }   \nonumber\\&
 \times \sqrt{ \sqrt{ 1+\left(c_2  v_{k1n} ^2+c_3  a_{k1n} v_{k1n} \right)^2  + {c_4^2 v_{k1n}^4} } - c_4  v_{k1n}^2   }   \nonumber\\& +  c_5   v_{k1n} ^3\Bigg].
\end{align}
Equation \eqref{pp1min1_2_1eq1} is derived from \eqref{apArweq12}.

Problem (\ref{pp1min1_2_1}) is non-convex  due to the objective function
\eqref{pp1min1_2_1eq1}.
It is generally hard to find the globally optimal solution of non-convex problem \eqref{pp1min1_2_1}.
In the following, we use the successive convex approximation (SCA) technique to obtain a suboptimal solution.

The second term in \eqref{pp1min1_2_1eq1} is non-convex.
To handle this issue, we introduce slack variables $A_{n}$ and $B_n$ such that
 \begin{align}\label{pp1min1_2_1eq2}
A_n= c_2  v_{k1n} ^2+c_3  a_{k1n} v_{k1n},
\end{align}
and
\begin{align}\label{pp1min1_2_1eq3}
B_n= \sqrt{ \sqrt{ 1+A_n^2  + {c_4^2 v_{k1n}^4} } + c_4  v_{k1n}^2   }.
\end{align}
Therefore, the second term in \eqref{pp1min1_2_1eq1} can be replaced by the convex expression $\frac{1+A_n^2}{B_n}$.
With the above manipulations,  problem \eqref{pp1min1_2_1} can be equivalent to
\begin{subequations}\label{pp1min1_2_2}
\begin{align}
\mathop{\min}_{ \substack{ \{q_{k1n}\},  \{v_{k1n}\}, \\ \{a_{k1n}\},\{A_n\},\{B_n\}}}
&  \sum_{n=1}^N \!\Bigg[\!P_0\! \left(\!1\!+\! c_1 v_{k1n}^2\! \right)   \! + \!\frac{ P_1(1\!+\!A_n^2) }{B_n}\! + \! c_5   v_{k1n} ^3\!\Bigg]
  \\
\textrm{s.t.}\qquad \:\:
& A_n\geq c_2  v_{k1n} ^2+c_3  a_{k1n} v_{k1n},   \forall n \in \mathcal N    \\
& B_n^2 \leq { \sqrt{ 1\!+\!A_n^2  \!+\! {c_4^2 v_{k1n}^4} } \!+ \! c_4  v_{k1n}^2   }, \forall n \in \mathcal N
 \\
& B_n\geq 0, \forall n \in \mathcal N\\
&(\ref{pp1min1_2_1}b)-(\ref{pp1min1_2_1}c), (\ref{pp1min1_2eq2})-(\ref{pp1min1_2eq3}).
\end{align}
\end{subequations}
Note that the constraints (\ref{pp1min1_2_2}b) and (\ref{pp1min1_2_2}c) are obtained from \eqref{pp1min1_2_1eq2} and \eqref{pp1min1_2_1eq3} by replacing the equalities with inequalities.
The reason is that for the optimal solution of problem (\ref{pp1min1_2_2}),
constraints (\ref{pp1min1_2_2}b) and (\ref{pp1min1_2_2}c) always hold with equality.

Problem (\ref{pp1min1_2_2}) is still non-convex due to the introduced non-convex constraints (\ref{pp1min1_2_2}b) and (\ref{pp1min1_2_2}c).
For the non-convex term $a_{k1n}v_{k1n}$ in constraints (\ref{pp1min1_2_2}b),  we have
\begin{align}\label{pp1min1_2_2eq1}
a_{k1n}v_{k1n}&=\frac1 4[(a_{k1n}+v_{k1n})^2-(a_{k1n}-v_{k1n})^2]
\nonumber\\&
\leq \frac1 4[(a_{k1n}+v_{k1n})^2 -2(a_{k1n}^{(j)}-v_{k1n}^{(j)})
\nonumber\\&\quad
\times(a_{k1n}-a_{k1n}^{(j)}+v_{k1n}-v_{k1n}^{(j)})-(a_{k1n}^{(j)}-v_{k1n}^{(j)})^2]
\nonumber\\&
\triangleq r_{1n}^{(j)}(a_{k1n},v_{k1n}),
\end{align}
where the superscript $(j)$ means the value of variable at the $j$-th iteration and the inequality follows from the fact that a convex function is no less than its first-order Taylor expansion.

Through analyzing all the principal minors of Hessian matrix,  $\sqrt{ 1+A_n^2  + {c_4^2 v_{k1n}^4}}$ is a convex function of variable $(A_n,v_{k1n})$.
For the right hand side of constraints (\ref{pp1min1_2_2}c),  we can obtain
\begin{align}\label{pp1min1_2_2eq2}
&\sqrt{ 1\!+\!A_n^2  \!+\! {c_4^2 v_{k1n}^4} } \!+ \! c_4  v_{k1n}^2  \geq  ( 1\!+\!(A_n^{(j)})^2  \!+\! {c_4^2 (v_{k1n}^{(j)})^4} )^{-\frac 1 2}
\nonumber\\ &\quad\quad\quad \times[A_n^{(j)}(A_n-A_n^{(j)})+
2c_4^2 (v_{k1n}^{(j)})^3 (v_{k1n}-v_{k1n}^{(j)})
]
\nonumber\\&\quad\quad\quad
+\sqrt{1\!+\!(A_n^{(j)})^2  \!+\! {c_4^2 (v_{k1n}^{(j)})^4}}+c_4  (v_{k1n}^{(j)})^2
\nonumber\\&\quad\quad\quad
 +2c_4(v_{k1n}-v_{k1n}^{(j)})\triangleq r_{2n}^{(j)}(A_n,v_{k1n}).
\end{align}

By replacing the term $a_{k1n}v_{k1n}$ in (\ref{pp1min1_2_2}b) with its upper bound $r_{1n}^{(j)}(a_{k1n},v_{k1n})$ and the right hand side of  constraints (\ref{pp1min1_2_2}c) with the lower bound $r_{2n}^{(j)} (A_n,v_{k1n})$, problem (\ref{pp1min1_2_2}) becomes a convex problem, which can be effectively solved by the interior point method~\cite{boyd2004convex}.
The SCA-based algorithm for problem (\ref{pp1min1_2_2}) is summarized in Algorithm 2.

\begin{algorithm}[t]
\caption{SCA-Based Algorithm for Problem (\ref{pp1min1_2_2})}
\begin{algorithmic}[1]
\STATE Obtain a feasible $(\{q_{k1n}^{(0)}\}, \{v_{k1n}^{(0)}\},  \{a_{k1n}^{(0)}\},\{A_n^{(0)}\},\{B_n^{(0)}\} )$ of problem (\ref{pp1min1_2_2}). Set $j=0$.
\REPEAT
\STATE
Replace  $a_{k1n}v_{k1n}$ in (\ref{pp1min1_2_2}b) with $r_{1n}^{(j)} (a_{k1n},v_{k1n})$ and the right hand side of (\ref{pp1min1_2_2}c) with $r_{2n}^{(j)}(A_n,v_{k1n})$.
\STATE Obtain the optimal $(\{\!q_{k1n}^{(j+1)}\!\}, \{\!v_{k1n}^{(j+1)}\!\},  \{\!a_{k1n}^{(j+1)}\!\},\{\!A_n^{(j+1)}\!\},$ $\{B_n^{(j+1)}\} )$ of  convex problem (\ref{pp1min1_2_2}).
\STATE Set $j=j+1$.
\UNTIL the objective value (\ref{pp1min1_2_2}a) converges.
\end{algorithmic}
\end{algorithm}
}}

\section{Energy Minimization with Fixed Path Planning}
In this section, the energy minimization problems for HD and FD modes with optimal path planning are respectively solved.
According to problems (\ref{sys1min1}) and (\ref{sys1min1_2}), the energy minimization with fixed path planning can be formulated as
 \begin{subequations}\label{ee2min1}
\begin{align}
\mathop{\min}_{\tilde{ \pmb q}, \tilde{\pmb t}, \pmb h, \pmb \rho, \pmb p   }\!\!&\:
\sum_{k=1}^{K} [E_2(q_{k2}(t),t_{k2})\!+\!E_4(q_{k4}(t), t_{k4})
 \!+\!(P_0\!+\!P_1)t_{k3}]  \\
\textrm{s.t.}\:\:\:
 &r_{\pi_k}^{\text {HD}} \geq D_{\pi_k},  \forall k \in\mathcal K \\
& e_{\pi_k}^{\text {HD}} \geq c_{\pi_k}^{\text {HD}},   \forall k \in\mathcal K\\
& 0\leq \rho_k\leq 1,   \forall k \in\mathcal K\\
& q_{ks}(0)\!=\!0,  q_{ks}(t_{ks})\!=\!H\!-\!h_k, \forall k\! \in\!\mathcal K\!, s\!=\!2,4 \\
&v_{ks}(0)=v_{ks}(t_{ks})=0,  \forall k \in\mathcal K, s\!=\!2,4\\
&  v_{ks}(t)\!=\!\dot{q}_{ks}(t), a_{ks}(t)\!=\!\ddot{q}_{ks}(t),\! \forall k \!\in\!\mathcal K, \!s\!=\!2,4\!\!\\
&|v_{ks}(t)|\!\leq\! V_{\max},|a_{ks}(t)|\!\leq\! A_{\max}, \!\forall k  \!\in\!\mathcal K , s\!=\!2,4,\!\!\!
\end{align}
\end{subequations}
for HD,
and
 \begin{subequations}\label{ee2min1_2}
\begin{align}
\mathop{\min}_{ \tilde{\pmb q}, \tilde{\pmb t}, \pmb h, \pmb p   }\quad
& \sum_{k=1}^{K} [E_2(q_{k2}(t),t_{k2})\!+\!E_4(q_{k4}(t), t_{k4})
 \!+\!(P_0\!+\!P_1)t_{k3}]\\
\textrm{s.t.}\quad
& r_{\pi_k}^{\text {FD}} \geq D_{\pi_k},  \forall k \in\mathcal K \\
& e_{\pi_k}^{\text {FD}} \geq c_{\pi_k}^{\text {HD}},   \forall k \in\mathcal K\\
& (\ref{ee2min1}e)-(\ref{ee2min1}h),
\end{align}
\end{subequations}
for FD.

From \eqref{ee2min1} and \eqref{ee2min1_2}, both objective function and constraints can be decoupled for $k$. Thus, it reduces to solve the specific energy minimization problem for a specific user in steps 2-4 of each stage.
It is also observed that with fixed $h_k$ the energy consumption $E_2(q_{k2}(t),t_{k2})+E_4(q_{k4}(t), t_{k4})$ is independent of energy consumption $(P_0+P_1)t_{k3}$.
Based on the above observations, problem \eqref{ee2min1} for  a specific $k$ can be equivalent to
 \begin{equation}\label{ee2min2_0}
 \mathop{\min}_{h_k}=V_1(h_k) +V_2(h_k),
\end{equation}
where $V_1(h_k)$ is the optimal objective value of problem \eqref{ee2min2_1} given as
\begin{subequations}\label{ee2min2_1}
\begin{align}
\mathop{\min}_{\substack{q_{k2}(t),q_{k4}(t),\\ t_{k2}, t_{k4}} }\: &
 E_2(q_{k2}(t),t_{k2})+E_4(q_{k4}(t), t_{k4})  \\
\textrm{s.t.}\:  \:
& q_{ks}(0)\!=\!0,  q_{ks}(t_{ks})\!=\!H-h_k, \forall s\!=\!2,4 \\
&v_{ks}(0)\!=\!v_{ks}(t_{ks})\!=\!0,  \forall  s\!=\!2,4\\
&  v_{ks}(t)\!=\!\dot{q}_{ks}(t), a_{ks}(t)\!=\!\ddot{q}_{ks}(t), \forall s\!=\!2,4\\
&|v_{ks}(t)|\!\leq\! V_{\max},|a_{ks}(t)|\!\leq\! A_{\max}, \forall  s\!=\!2,4,
\end{align}
\end{subequations}
and $V_2(h_k)$ is the optimal objective value of problem \eqref{ee2min2_2} given as
\begin{subequations}\label{ee2min2_2}
\begin{align}
\mathop{\min}_{t_{k3},\rho_k,p_{\pi_k}}\: &
(P_0+P_1)t_{k3} \\
\textrm{s.t.}\quad  \:
 &(1-\rho_k)
t_{k3} B\log_2
\left( 1 + \frac {p_{\pi_k} b  h_k^{-\alpha}}{B \sigma^2}
\right) \!\geq\!D_{\pi_k}\\
& \zeta  P b  h_k^{-\alpha} \rho_k t_{k3} \geq \rho_k t_{k3} p_{\pi_k}^{\text{re}}
\nonumber\\
&\qquad\qquad\qquad
+(1-\rho_k)t_{k3} \left(p_{\pi_k}^{\text{tr}}+ \frac{p_{\pi_k} }{\epsilon_{\pi_k}}\right)\\
& 0\leq \rho_k\leq 1.
\end{align}
\end{subequations}
Note that we have substituted the value of $r_{\pi_k}^{\text {FD}}, e_{\pi_k}^{\text {FD}}, c_{\pi_k}^{\text {HD}}$ into problem \eqref{ee2min2_2}  according to \eqref{sys1sf1eq3}-\eqref{sys1sf1eq3_1}.
Similarly, problem \eqref{ee2min1_2} with a fixed $k$ can be equivalent to
 \begin{equation}\label{ee2min3_0}
 \mathop{\min}_{h_k}=V_1(h_k) +V_3(h_k),
\end{equation}
where  $V_3(h_k)$ is the optimal objective value of problem   \eqref{ee2min3_2} given as
\begin{subequations}\label{ee2min3_2}
\begin{align}
\mathop{\min}_{t_{k3},p_{\pi_k}}\: &
(P_0+P_1)t_{k3} \\
\textrm{s.t.}\quad  \:
 &(t_{k3}-\delta_{\pi_k}) B \log_2 \left(
1+\frac{p_{\pi_k} b  h_k^{-\alpha}}{\gamma P +B\sigma^2}
\right) \!\geq\!D_{\pi_k}\\
& \zeta  P b  h_k^{-\alpha}t_{k3}
\geq
t_{k3}  p_{\pi_k}^{\text{re}}+ (t_{k3}-\delta_{\pi_k}) \left(p_{\pi_k}^{\text{tr}}+ \frac{p_{\pi_k} }{\epsilon_{\pi_k}}\right).
\end{align}
\end{subequations}

Since problems \eqref{ee2min2_0} and \eqref{ee2min3_0} only involve only one variable, the 1D exhaustive search method can be used to obtain the optimal solution.
In the following, we separately solve problems \eqref{ee2min2_1}, \eqref{ee2min2_2}, and \eqref{ee2min3_2}.

{\color{myc1}{
For the feasibility  of problem \eqref{ee2min2_1}, the feasible solution can be obtained as in \eqref{pp1min1_2feq1} and \eqref{pp1min1_2feq2}.
For problems \eqref{ee2min2_2} and \eqref{ee2min3_2}, the optimal solutions are obtained in closed form according to the following Theorems 2 and 3.

\subsection{Vertical Trajectory Optimization for Problem \eqref{ee2min2_1}}

According to Lemmas 1 and 2 in Section III-B, problem \eqref{ee2min2_1} can be effectively solved without loss of optimality even with fixed $t_{k2}$ and $t_{k4}$.
Similar to Section III-B, the discrete linear state-space approximation is adopted. By discretizing
the time duration $t_{k2}$ ($t_{k4}$) into $N+1$ slots with step size $\delta_2=\frac{t_{k2}}{N_2}$ ($\delta_4=\frac{t_{k4}}{N_4}$), i.e., $t=n\delta_2$ ($t=n\delta_4$), $n=0,1,\cdots,N$, the UAV trajectory $q_{k2}(t)$ ($q_{k4}(t)$) can be characterized by the discrete-time UAV location $\{q_{k4n}\}$ ($\{q_{k4n}\}$).
As a result, problem \eqref{ee2min2_1} can be expressed in the discrete form as
\begin{subequations}\label{ee2min2_1_1}
\begin{align}
\mathop{\min}_{ \substack{ \{q_{k2n}\},\{v_{k2n}\},\{a_{k2n}\},\\\{q_{k4n}\},\{v_{k4n}\},\{a_{k4n}\}}}
&  P_2(t_{k2}+t_{k4}) + W h_k
\nonumber\\
+\sum_{n=1}^N&
\frac {W-ma_{k2n}} 2 \sqrt{v_{k2n}^2 + \frac{2(W-ma_{k2n})}{\rho A}}
\nonumber\\
+\sum_{n=1}^N&
\frac {W+ma_{k4n}} 2 \sqrt{v_{k4n}^2 + \frac{2(W+ma_{k4n})}{\rho A}}
 \\
\textrm{s.t.}\quad
& q_{ks0}\!=\!0,  q_{ksN}\!=\!H\!-\!h_k, \forall s\!=\!2,4 \\
&v_{ks0}=v_{ksN}=0,  \forall  s\!=\!2,4\\
&  (\ref{pp1min1_2eq2})-(\ref{pp1min1_2eq3}),
\end{align}
\end{subequations}
The formulation of objective function (\ref{ee2min2_1_1}a) comes from
\eqref{apArweq23}  and \eqref{apArweq24}.

Problem \eqref{ee2min2_1_1}  is non-convex due to non-convex objective function (\ref{ee2min2_1_1}a).
To handle this issue, we introduce variables $X_n$ and $Y_n$ such that
 \begin{align}\label{ee2min2_1_1eq2}
X_n= \sqrt{v_{k2n}^2 + \frac{2(W-ma_{k2n})}{\rho A}}    ,
\end{align}
and
\begin{align}\label{ee2min2_1_1eq3}
Y_n= \sqrt{v_{k4n}^2 + \frac{2(W+ma_{k4n})}{\rho A}}   .
\end{align}
With the above manipulations,  problem \eqref{ee2min2_1_1} can be equivalent to
\begin{subequations}\label{ee2min2_1_1_2}
\begin{align}
\mathop{\min}_{ \substack{ \{q_{k2n}\},\{v_{k2n}\},\{a_{k2n}\},\\\{q_{k4n}\},\{v_{k4n}\},\{a_{k4n}\},\\\{X_n\},\{Y_n\}
}}
&  P_2(t_{k2}+t_{k4}) + W h_k
\nonumber\\
+\sum_{n=1}^N&
\frac {WX_n-ma_{k2n}X_n+WY_n+ma_{k4n}Y_n} 2
 \\
\textrm{s.t.}\quad
& X_n^2 \geq {v_{k2n}^2 + \frac{2(W-ma_{k2n})}{\rho A}}
 \\
 &Y_n^2 \geq {v_{k4n}^2 + \frac{2(W+ma_{k4n})}{\rho A}}   \\\
&  (\ref{ee2min2_1_1}b), (\ref{ee2min2_1_1}c), (\ref{pp1min1_2eq2})-(\ref{pp1min1_2eq3}).
\end{align}
\end{subequations}
Note that the constraints (\ref{ee2min2_1_1_2}b) and (\ref{ee2min2_1_1_2}c) are obtained from \eqref{ee2min2_1_1eq2} and \eqref{ee2min2_1_1eq3} by replacing the equalities with inequalities.
To handle the non-convexity of $-a_{k2n}X_n+a_{k4n}Y_n$ in (\ref{pp1min1_2_2}a), the
convex approximation can be obtained by using the similar method  in \eqref{pp1min1_2_2eq1}. Moreover, for the left hand sides of  (\ref{ee2min2_1_1_2}b) and (\ref{ee2min2_1_1_2}c) can be approximated by its first-order Taylor expansion as in \eqref{pp1min1_2_2eq2}.
As a result, problem \eqref{ee2min2_1_1_2} can be solved by the SCA-based algorithm. }}

\subsection{HD Mode for Problem \eqref{ee2min2_2}}
\begin{theorem}
For HD, the optimal solution of problem \eqref{ee2min2_2}, $(t_{k3}^*, \rho_k ^*, p_{\pi_k}^*)$, is given by
\begin{equation}\label{sf2hdeq5_0}
t_{k3}^*=t_{k31}^*+t_{k32}^*, \rho_k^*=\frac{t_{k31}^*}{t_{k31}^*+t_{k32}^*},
\end{equation}
and
\begin{equation}\label{sf2hdeq5_0_2}
p_{\pi_k}^*=\frac{\epsilon_{\pi_k}( \zeta  P b  h_k^{-\alpha} -p_{\pi_k}^{\text{re}} )t_{k31}^*}{t_{k32}^*} -\epsilon_{\pi_k}p_{\pi_k}^{\text{tr}}.,
\end{equation}
where
 \begin{equation}\label{sf2hdeq5_2}
 t_{k31}^* = \left( 2^{\frac{D}{B t_{k32}^* }} -u_1 \right) \frac{t_{k32}^*  }{u_2}.
\end{equation}
\begin{equation}\label{sf2hdeq5_1}
t_{k32}^*= \frac{ (\ln 2)D}{ B W \left(  \frac{u_2-1}
{\text e  }\right) +B},
\end{equation}
 \begin{equation}\label{sf2hdeq5_2_2}
u_1=1-\frac {\epsilon_{\pi_k}p_{\pi_k}^{\text{tr}} b  h_k^{-\alpha}}{B \sigma^2},u_2= \frac {\epsilon_{\pi_k}( \zeta  P b  h_k^{-\alpha} -p_{\pi_k}^{\text{re}} ) b  h_k^{-\alpha}}{B \sigma^2}.
\end{equation}
and
 $W(\cdot)$ is the Lambert-W function.

\end{theorem}

\begin{proof}
Please refer to Appendix C.
\end{proof}

According to Theorem 1, the optimal solution of problem \eqref{ee2min2_2} for HD is obtained in closed form.
{\color{myc1}{Based on \eqref{sf2hdeq5_2_2}, we must have $\zeta  P b  h_k^{-\alpha} -p_{\pi_k}^{\text{re}}>0$, which means that the harvested power at user $\pi_k$ must be greater than its receive power, i.e., the following constraint about the UAV hight is obtained
\begin{equation}\label{sf2hdeq3_5}
h_k < \left( \frac{\zeta  P b }{p_{\pi_k}^{\text{re}}}
\right)^{\frac 1{\alpha}}.
\end{equation}}}


\vspace{-1em}
\subsection{FD Mode for Problem \eqref{ee2min3_2}}
\begin{theorem}
For FD,  the optimal solution of problem \eqref{ee2min3_2}, $(t_{k3}^*,  p_{\pi_k}^*)$, satisfies,
$t_{k3}^*$ is the solution to the following equation
\begin{align}\label{sf2fdeq3_1}
\left( 2^{\frac{D}{B (t_{k3} -\delta_{\pi_k})}} -1 \right) \frac{\gamma P +B\sigma^2}{b  h_k^{-\alpha}}&-\frac{\epsilon_{\pi_k}( \zeta  P b  h_k^{-\alpha} -p_{\pi_k}^{\text{re}} )t_{k3}}{t_{k3}-\delta_{\pi_k}}
\nonumber\\& +\epsilon_{\pi_k}p_{\pi_k}^{\text{tr}}=0.
\end{align}
and
\begin{equation}\label{sf2fdeq3_2}
p_{\pi_k}^*=\left( 2^{\frac{D}{B (t_{k3}^* -\delta_{\pi_k})}} -1 \right) \frac{\gamma P +B\sigma^2}{b  h_k^{-\alpha}}.
\end{equation}

\end{theorem}

\begin{proof}
Please refer to Appendix D.
\end{proof}

Since the left hand side of equation \eqref{sf2fdeq3_1} is monotonically decreasing with $t_{k3}$, the unique solution $t_{k3}^*$ to  \eqref{sf2fdeq3_1} can be obtained by the bisection method.
For the case that $\delta_k\approx0$, to ensure $p_{\pi_k}>0$ from (\ref{ee2min3_2}c), we must have
\begin{equation}\label{sf2fdeq3_5}
h_k < \left( \frac{\zeta  P b }{p_{\pi_k}^{\text{re}}+p_{\pi_k}^{\text{tr}}}
\right)^{\frac 1{\alpha}}.
\end{equation}

\section{Numerical Results}

There are $K=8$ users uniformly in a square area of size $1000$ m $\times$ $1000$~m with the UAV initially located at its center.
The noise power spectrum density is $\sigma^2=-174$ dBm/Hz.
We set reference channel gain $\beta_0=1.42\times 10^{-4}$ \cite{Yang2019OptimizationOR} and energy harvesting coefficient $\zeta=0.9$.
The maximal flying speed of the UAV is $V_{\max}=30$ m/s and the maximal acceleration is $A_{\max}=5$ m/s$^2$.
All users have the same power amplifier efficiency $\epsilon_k=0.9$, the same processing delay of the energy circuits $\delta_{k}=2$ s, and the same minimal required data $D_k=D$.
The modeling coefficients for the probabilistic LoS channel model are set as $C_1=10$, $C_2=0.6$, $\kappa=0.2$, and $\alpha=2.3$.
Unless specified otherwise, the system parameters are set as: {\color{myc1}{the UAV altitude is $H = 120$ m}}, the maximal wireless transmission power of the UAV is $P =1$~W, the system bandwidth is $B = 20$ MHz,
the receive and transmit circuit power consumption  are respectively set as $p_{k}^{\text {re}}=p^{\text {re}}=1$~uW and $p_{k}^{\text {tr}}=p^{\text {tr}}=1$~mW for all users, the self-interference coefficient $\gamma=-100$ dBm,
and the minimal required data is $D=1$ Mbits.

\begin{figure}
\centering
\includegraphics[width=3.45in]{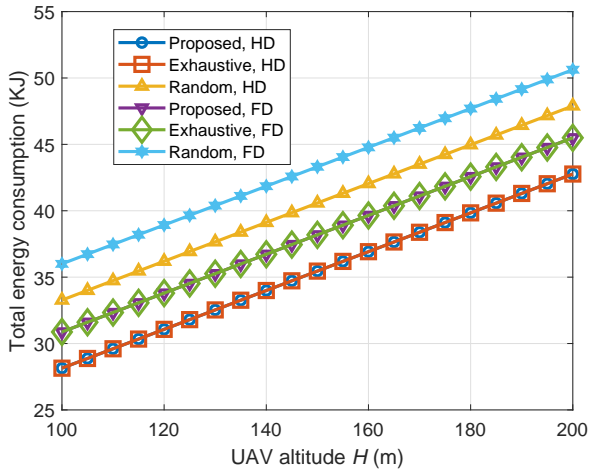}
\caption{Total energy consumption versus UAV altitude $H$.}\label{nr5fig3}
\end{figure}

The total energy consumption versus UAV altitude $H$ is illustrated in Fig.~\ref{nr5fig3}.
We compare the proposed optimal path planning Algorithm 1 (labeled as `Proposed'), with the exhaustive search method to obtain the optimal path planning (labeled as `Exhaustive'), and the random method to obtain a feasible path planning (labeled as `Random').
It is found that the proposed algorithm achieves the same performance as the exhaustive search method, which shows the optimality of the proposed algorithm.
From this figure, it is also shown that the total energy consumption linearly increases with altitude $H$.
This is because high altitude means long distance the UAV needs to change for energy broadcast and information reception, which results high propulsion energy consumption of the UAV.

{\color{myc1}{
\begin{figure}
\centering
\includegraphics[width=3.45in]{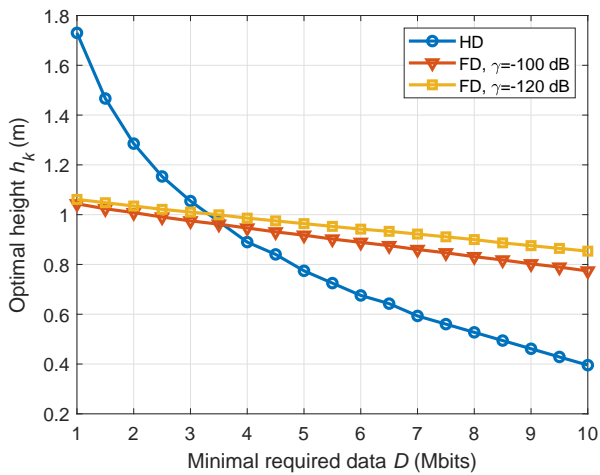}
\caption{Optimal height versus minimal required data $D$.}\label{nr5fig7}
\end{figure}

In Fig.~\ref{nr5fig7}, we  illustrate the optimal height $h_k$ in problems \eqref{ee2min2_0} and \eqref{ee2min3_0} versus minimal required data $D$.
From this figure, it is found that the optimal height decreases with an increasing $D$.
This is due to the fact that high minimal required data needs high data rate to decrease the transmission time, while low height can lead to high data rate between the user and the UAV.
It is also found that the decrease speed of the optimal height for HD is higher than that for FD.
For low data demand, i.e., $D\leq 3$ Mbits, the optimal hight of HD is higher than that of FD.
As for high data demand, i.e., $D> 3$ Mbits, the optimal hight of HD is lower than that of FD.
According to this figure, it is observed that the optimal hight is low, less than 2 m.
The reason is that the channel gain is the better for small hight, which results in short staying time and low energy consumption.

\begin{figure}
\centering
\includegraphics[width=3.45in]{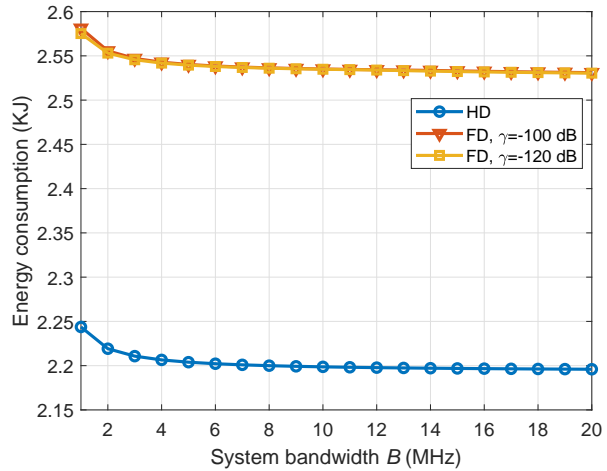}
\caption{Energy consumption versus system bandwidth with $D=10$ Mbits.}\label{nr5fig4}
\end{figure}

The energy consumption in the following Figs. \ref{nr5fig4} to \ref{nr5fig9} means the energy consumption in (\ref{ee2min2_0}a) and (\ref{ee2min3_0}a), which includes the energy for the UAV to firstly decrease the altitude, then collect data and finally increase the altitude.
Fig. \ref{nr5fig4} presents the energy consumption versus system bandwidth.
From this figure, it is shown that the energy consumption first decreases rapidly with bandwidth and then decreases slowly with bandwidth.
This is because for small bandwidth, the UAV hovering time is long and the propulsion energy during hovering decreases rapidly with the increase of bandwidth.
For high bandwidth, the hovering time is short and thus the propulsion energy for the UAV to change height dominates the energy consumption.}}

\begin{figure}
\centering
\includegraphics[width=3.45in]{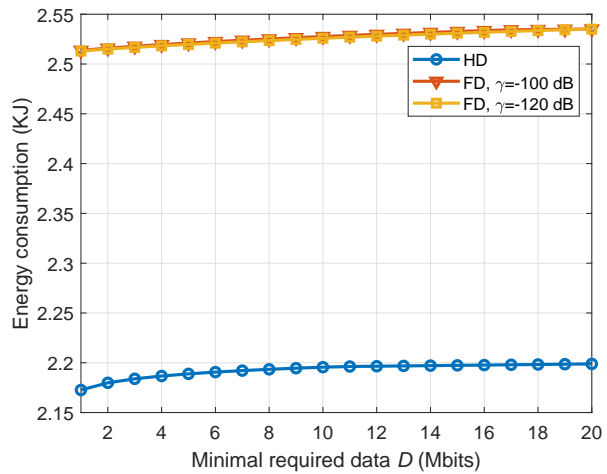}
\caption{Energy consumption versus minimal required data $D$.}\label{nr5fig8}
\end{figure}

The total energy consumption versus minimal data demand $D$ is depicted in Fig.~\ref{nr5fig8}.
It can be observed that the total energy consumption monotonically increases with minimal data demand.
This is because high minimal data demand leads to long hovering time of the UAV, which increases the energy consumption of the UAV.
It is also found that the HD is always superior over FD, and FD with low self-interference coefficient, i.e., $\gamma=-120$ dB, yields slightly better performance when the self-interference coefficient is high, i.e., $\gamma=-100$ dB.

\begin{figure}
\centering
\includegraphics[width=3.45in]{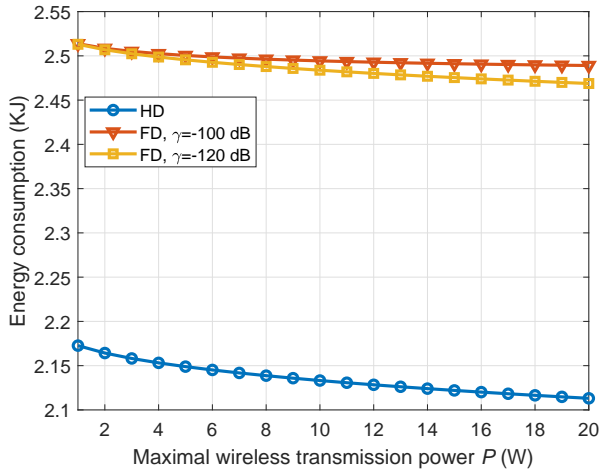}
\caption{Energy consumption versus maximal wireless transmission power $P$.}\label{nr5fig9}
\end{figure}

Fig.~\ref{nr5fig9} presents the total energy consumption versus maximal wireless transmission power $P$.
We find that the total energy consumption first decreases dramatically and then decreases smooth with maximal wireless transmission power $P$.
According to Figs. \ref{nr5fig4} to \ref{nr5fig9}, it is found that HD is always better than FD.
The reason is that the transmit circuit power consumption of the user is high in FD, while there is only small receive circuit power consumption of the user in the energy harvesting stage of HD mode.
As a result, the transmit power of the user in FD is lower than that in HD, which results in longer hovering time and higher energy consumption of the UAV.

%

\section{Conclusion}
In this paper, we have investigated the total energy minimization problem for UAV communication with energy harvesting.
It is shown that the UAV should stay directly above the user with low height for energy transferring and information reception.
It is also found that HD mode is recommended compared to FD mode.
{\color{myc1}{The general 3D trajectory optimization for UAV communication with beamforming, instability error, and delay requirements is left for future work.}}


\appendices
{\color{myc1}{
\section{Rotary-Wing UAV Propulsion Energy Consumption}
\setcounter{equation}{0}
\renewcommand{\theequation}{\thesection.\arabic{equation}}

In this appendix, we derive the propulsion energy consumption
model of rotary-wing UAVs.
The main notations and typical values used in this appendix are summarized in Table~I.

\begin{table}[h]\footnotesize
 \caption{List of main notations and typical values.}
\begin{tabular}{c|l|l}
\hline
  Notation & Physical meaning & Value \\ \hline
  $W$ &   UAV weight in Newton       & 20 \\ \hline
  $\rho$ &   Air density in kg/m$^3$       & 1.225 \\ \hline
  $S_{\text{FP}}$ & Fuselage equivalent flat plate area in m$^2$       & 0.0151 \\ \hline
  $R$ & Rotor radius in meter (m)       & 0.4 \\ \hline
   $A$ & Rotor disc area in m$^2$       & 0.503 \\ \hline
   $\Omega$ & Blade angular velocity in radians/second & 300 \\\hline
      $d_0$ & Fuselage drag ratio & 0.6
   \\ \hline
   $s$ & Rotor solidity & 0.05 \\\hline
      $\delta$ & Profile drag coefficient & 0.012 \\ \hline
   $k$ & Incremental correction factor to induced power & 0.1
   \\ \hline
      $m$ & UAV mass in kg   & 2.04   \\ \hline
         $g$ & Gravitational acceleration in m/s$^2$  & 9.8 \\ \hline
\end{tabular}
\end{table}

\subsection{Rotary-Wing UAVs in Forward Straight Flight}

\begin{figure}
\centering
\includegraphics[width=3.5in]{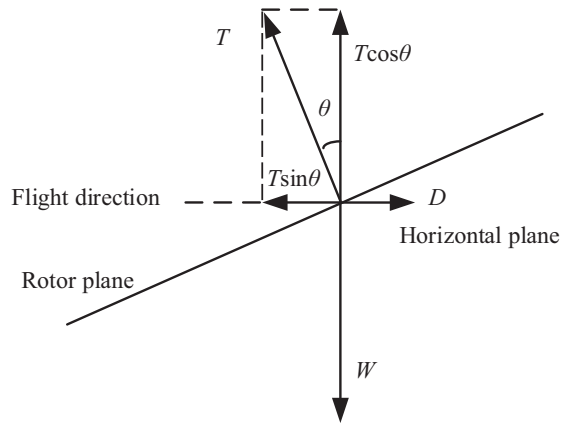}
\vspace{-2em}
\caption{Schematic of the forces on a UAV with a fixed height and straight flight.}\label{sys1figu}
\vspace{-1.5em}
\end{figure}

We first obtain the  propulsion energy consumption
model with a fixed height and straight flight.
Fig.~\ref{sys1figu} shows the simplified schematics of the
longitudinal forces acting on the aircraft with fixed a height \cite[Fig. 13.2]{filippone2006flight}, which include the following forces:
(i) $T$: rotor thrust, normal to the disc plane and directed
upward; (ii) $D$: fuselage drag, which is in the opposite
direction of the aircraft velocity;
and (iii) $W$: aircraft
weight.
In Fig.~\ref{sys1figu}, $\theta$ is the tilt angle of the rotor disc.
From Fig.~\ref{sys1figu}, we have the following equation:
\begin{equation}\label{apArweq1}
T\sin\theta-D=m a,\quad W-T\cos\theta=0,
\end{equation}
where $a$ denotes the acceleration.
According to \cite[Eq. (4.5)]{bramwell2001bramwell}, the UAV fuselage drag $D$ can be written as
\begin{equation}\label{apArweq2}
D=\frac 1 2 \rho S_{\text{FP}}  V^2.
\end{equation}

Due to the complexity of deriving the precise power consumption for a rotary-wing aircraft, we assume that the drag coefficient of the blade section is constant \cite{zeng2019energy}.
 Under some mild assumptions, the propulsion power for rotary-wing UAV with fixed height, speed $V$ and rotor thrust $T$ can be given by  \cite[Eq. (66)]{zeng2019energy},
\begin{align}\label{apArweq3}
&P_{\text{f}}(V, \kappa)=P_0 \left(1+\frac{3V^2}{\Omega^2R^2}\right) \nonumber\\
&+ P_1  \kappa \sqrt{   \sqrt{\kappa^2 +\frac{\rho^2 A^2V^4}{W^2}} - \frac{\rho AV^2}{W}   } + \frac 1 2  d_0 \rho sA V^3,
\end{align}
where
\begin{equation}\label{apArweq5}
\kappa\triangleq \frac T W
\end{equation}
 is defined as the thrust-to-weight ratio,
 $P_0=\frac{\delta}{8}\rho sA \Omega^3 R^3$ and
 $P_1=(1+k)\frac{W^{\frac 3 2 }}{\sqrt{2\rho A}}$.


Based on \eqref{apArweq1}, \eqref{apArweq2} and \eqref{apArweq5}, we can obtain
\begin{equation}\label{apArweq6}
\kappa=\sqrt{1+\frac{\left(  \rho S_{\text{FP}} V^2 + 2ma \right)^2}{4W^2}}.
\end{equation}
Substituting \eqref{apArweq6} into \eqref{apArweq3} yields
\begin{align}\label{apArweq7}
P_{\text{f}}(V, &a)=P_0 \left(1+ c_1V^2 \right)+ P_1  \sqrt{ 1+(c_2V^2+c_3a )^2}   \nonumber\\
& \times \sqrt{ \sqrt{(1+(c_2V^2+c_3a )^2 + {c_4^2V^4} } - c_4 V^2   }+  c_5 V^3.
\end{align}
where
\begin{equation}\label{apArweq8}
c_1=\frac{3 }{\Omega^2R^2},
c_2=\frac{\rho S_{\text{FP}}}{2W},
c_3=\frac{m}{W},
c_4=\frac{\rho A}{W},
c_5=\frac 1 2  d_0 \rho sA.
\end{equation}


For a UAV with trajectory $q(t)$, we have velocity $ v(t)=\dot {q}(t)$ and acceleration $ a(t)=\ddot {  q}(t)$.
The total propulsion energy can be expressed as \eqref{apArweq12}, as shown in the front of the next page, where $T_0$ is the time duration.
\newcounter{mytempeqncnt}
\begin{figure*}[!t]
\normalsize
\setcounter{mytempeqncnt}{\value{equation}}
\begin{align}\label{apArweq12}
E_1( q(t),T_0)=\int_{0}^{T_0} \Bigg[&P_0 \left(1+ c_1 v(t)^2 \right)+  P_1  \sqrt{  1+\left(c_2  v(t)^2+c_3 a(t)  v(t)   \right)^2  }   \nonumber\\
& \cdot \sqrt{ \sqrt{ 1+\left(c_2  v(t) ^2+c_3  a(t) v(t)  \right)^2  + {c_4^2 v(t)^4} } - c_4  v(t)^2   } +  c_5   v(t) ^3\Bigg] \text d t.
\end{align}
\hrulefill
\vspace*{4pt}
\end{figure*}

\subsection{Rotary-Wing UAVs in Vertical Flight}

\begin{figure}
\centering
\includegraphics[width=3.5in]{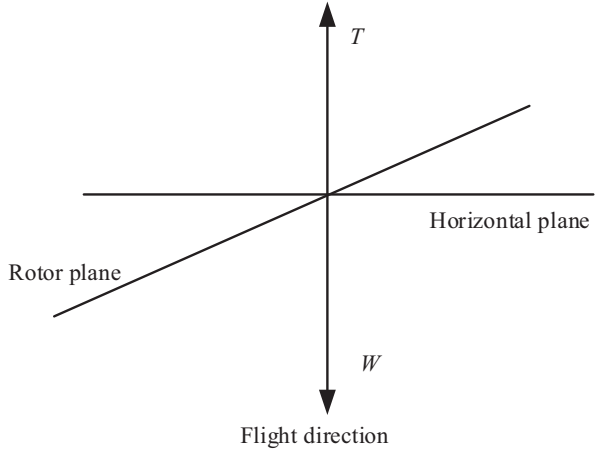}
\vspace{-2em}
\caption{Schematic of the forces on a UAV in vertical flight.}\label{sys1figu2}
\vspace{-1.5em}
\end{figure}

Then, we obtain the  propulsion energy consumption
model in vertical flight.
In Fig.~\ref{sys1figu2}, we show the simplified schematics of the
longitudinal forces acting on the aircraft with in vertical flight.
In vertical flight, it is assumed that the velocity is non-zero only in the vertical direction.

From Fig.~\ref{sys1figu2}, we have
\begin{equation}\label{apArweq21}
W-T=ma
\end{equation}
for vertical descend
and
\begin{equation}\label{apArweq21_2}
T-W=ma
\end{equation}
for vertical climb.

According to \cite[Eqs. (12.13), (12.35)]{filippone2006flight}, the total required power for vertical descend is
\begin{equation}\label{apArweq22}
P_{\text c}(V,T)=P_2+\frac{1}{2}TV + \frac T 2 \sqrt{V^2 + \frac{2T}{\rho A}},
\end{equation}
where $P_2=\frac{\delta}{8}\rho sA \Omega^3 R^3+\frac{kW^{\frac 3 2 }}{\sqrt{2\rho A}}$ and we assume that the acceleration $a$ is smaller than the Gravitational acceleration $g=9.8$ $\text{m/s}^2$.
Substituting $T=W-ma$ to \eqref{apArweq22} yields
\begin{equation}\label{apArweq22_2}
P_{\text c}(V,a)=P_2+\frac{W-ma}{2}  V + \frac {W-ma} 2 \sqrt{V^2 + \frac{2(W-ma)}{\rho A}}.
\end{equation}
In the vertical descend, the direction of both velocity and acceleration is vertical. The trajectory is one dimensional, which can be expressed by $q(t)$.
As a result, the total propulsion energy in vertical descend can be given y
\begin{align}\label{apArweq23}
E_2( q(t),T_0)&=\int_{0}^{T_0}\Bigg(P_2+\frac{W v(t)-ma(t)v(t)}{2}
\nonumber\\&
\quad + \frac {W-ma(t)} 2 \sqrt{v(t)^2 + \frac{2(W-ma(t))}{\rho A}} \Bigg) \text d t
\nonumber\\&=
P_2 T_0+\frac{W ( q(T_0)-q(0))}{2}
\nonumber\\&-\frac{m( v(T_0)^2-v(0)^2)}{4}+\int_{0}^{T_0} \frac {W-ma(t)} 2
\nonumber\\&\quad \times \sqrt{v(t)^2 + \frac{2(W-ma(t))}{\rho A}}   \text d t.
\end{align}

According to \cite[Eqs. (12.13), (12.51)]{filippone2006flight} and \eqref{apArweq21_2}, the required power for vertical climb can be calculated as
\begin{equation}
P_{\text d}(V,a)=P_2+\frac{W+ma}{2}  V + \frac {W+ma} 2 \sqrt{V^2 + \frac{2(W+ma)}{\rho A}}.
\end{equation}
Thus, the total propulsion energy in vertical climb can be given by
\begin{align}\label{apArweq24}
E_4( q(t),T_0)&=\int_{0}^{T_0}\Bigg(P_2+\frac{W v(t)+ma(t)v(t)}{2}
\nonumber\\&
\quad + \frac {W+ma(t)} 2 \sqrt{v(t)^2 + \frac{2(W+ma(t))}{\rho A}} \Bigg) \text d t
\nonumber\\&=
P_2 T_0+\frac{W ( q(T_0)-q(0))}{2}
\nonumber\\&+\frac{m( v(T_0)^2-v(0)^2)}{4}+\int_{0}^{T_0} \frac {W+ma(t)} 2
\nonumber\\&\quad \times \sqrt{v(t)^2 + \frac{2(W+ma(t))}{\rho A}}   \text d t.
\end{align}

\section{Proof of Theorem 1}
\setcounter{equation}{0}
\renewcommand{\theequation}{\thesection.\arabic{equation}}

The dual problem of problem (\ref{pp1min2_12}) with relaxed constraints can be given by:
\begin{equation}\label{dueq1}
\mathop{\max}_{\pmb \beta,\pmb \gamma,\pmb \lambda,\pmb \mu}\quad D(\pmb \beta,\pmb \gamma,\pmb \lambda,\pmb \mu),
\end{equation}
where
\begin{equation}\label{dueq1_2}
D( \pmb \beta,\pmb \gamma,\pmb \lambda,\pmb \mu)=\left\{ \begin{array}{ll}
\!\!\!\mathop{\min}\limits_{\pmb w, \pmb v}
&  \mathcal L (\pmb w, \pmb v, \pmb \beta,\pmb \gamma,\pmb \lambda,\pmb \mu)
\\
\!\!\rm{s.t.}&\sum_{k=1}^K w_{kl}=1,   \forall l\in\mathcal  K\\
& w_{kl}\in[0,1],  \forall k,l\in\mathcal K\\
&v_{kli}\!\geq \!0,  \forall k\!\in\!\mathcal K\!\setminus\!\{1\},l,i\!\in\!\mathcal K,\!\!\!\!\!
\end{array} \right.
\end{equation}
with
\begin{align}\label{dueq2}
&\mathcal L (\pmb w, \pmb v, \pmb \beta,\pmb \gamma,\pmb \lambda,\pmb \mu)=\sum_{l=1}^K w_{1l} E_{0l} + \sum_{k=2}^{K} \sum_{l=1}^K  \sum_{i=1}^K v_{kli}^2 E_{li}
\nonumber\\ & +\sum_{l=1}^K w_{Kl} E_{l0}  +\sum_{k=1}^K \beta_k \left(\sum_{l=1}^K w_{kl}-1\right)
\nonumber\\&
+\sum_{k=2}^{K} \sum_{l=1}^K  \sum_{i=1}^K \bigg[\gamma_{kli} ( w_{(k-1)l}+ w_{ki}-1- v_{kli} )
 \nonumber\\&\qquad\qquad+\lambda_{kli}(v_{kli}-w_{(k-1)l})  +\mu_{kli}(v_{kli}- w_{ki})\bigg]
\end{align}
and $\pmb \beta=\{\beta_k\}, \pmb \gamma=\{\gamma_{kli}\},\pmb\lambda=\{\lambda_{kli}\}, \pmb\mu=\{\mu_{kli}\}$.

To minimize the objective function in (\ref{dueq1_2}), which is a linear combination of $w_{kl}$, we should
let the smallest association coefficient corresponding to the $w_{kl}$   be 1 among all $k$ with given $l$.
Therefore, the optimal $w_{kl}^*$ is thus given as \eqref{PAuser2eq2_2}.

To obtain the optimal $v_{kli}^*$ from \eqref{dueq1_2}, we set the first derivative of objective function to zero, i.e.,
\begin{equation}\label{PAuser2eq2_6}
\frac{\partial\mathcal L (\pmb w, \pmb v, \pmb \beta,\pmb \gamma,\pmb \lambda,\pmb \mu)}
{\partial v_{kli}}=2E_{li} v_{kli} -\gamma_{kli} +\lambda_{kli} +\mu_{kli}=0,
\end{equation}
which yields $v_{kli}=\frac{\gamma_{kli} -\lambda_{kli} -\mu_{kli}}{2E_{li}}$.
Considering constraint $v_{kli} \geq 0$, we can obtain the optimal solution to problem (\ref{pp1min2_12}) as \eqref{PAuser2eq2_6_2}.

}}

\section{Proof of Theorem 2}
\setcounter{equation}{0}
\renewcommand{\theequation}{\thesection.\arabic{equation}}
Introducing  $t_{k31}=\rho_k t_{k3}$ and $t_{k32}=(1-\rho_k) t_{k32}$,
problem \eqref{ee2min2_2} can be equivalent to
\begin{subequations}\label{sf2hdmin1}
\begin{align}
\mathop{\min}_{t_{k31},t_{k32},p_{\pi_k}}\: &
 t_{k31}+t_{k32} \\
\textrm{s.t.}\quad  \:
 &
t_{k32} B\log_2
\left( 1 + \frac {p_{\pi_k} b  h_k^{-\alpha}}{B \sigma^2}
\right) \!\geq\!D_{\pi_k}\\
& \zeta  P b  h_k^{-\alpha}  t_{k31} \geq t_{k31} p_{\pi_k}^{\text{re}}
+ t_{k32} \left(p_{\pi_k}^{\text{tr}}+ \frac{p_{\pi_k} }{\epsilon_{\pi_k}}\right).
\end{align}
\end{subequations}

Observing  \eqref{sf2hdmin1}, constraint (\ref{sf2hdmin1}c) holds with equality for the optimal solution.
Based on constraint (\ref{sf2hdmin1}c) with equality, we can obtain
\begin{equation}\label{sf2hdeq1}
p_{\pi_k}^*=\frac{\epsilon_{\pi_k}( \zeta  P b  h_k^{-\alpha} -p_{\pi_k}^{\text{re}} )t_{k31}}{t_{k32}} -\epsilon_{\pi_k}p_{\pi_k}^{\text{tr}}.
\end{equation}

Substituting \eqref{sf2hdeq1} into problem \eqref{sf2hdmin1} yields
\begin{subequations}\label{sf2hdmin2}
\begin{align}
\mathop{\min}_{ t_{k31},t_{k32}\geq 0}\quad
&t_{k31}+t_{k32} \\
\textrm{s.t.}\qquad
& t_{k32} B\log_2
\left( u_1 + \frac {u_2 t_{k31}}{t_{k32}}
\right) \!\geq\!D_{\pi_k},
\end{align}
\end{subequations}
where $u_1$ and $u_2$ are defined in \eqref{sf2hdeq5_2_2}.
For the optimal solution, constraint (\ref{sf2hdmin2}b) always holds with equality,
which yields
\begin{equation}\label{sf2hdeq2}
t_{k31}^* = \left( 2^{\frac{D_{\pi_k}}{B t_{k32} }} -u_1 \right) \frac{t_{k32}  }
{u_2}.
\end{equation}
Using \eqref{sf2hdeq2}, problem \eqref{sf2hdmin2} is simplified to
\begin{equation}\label{sf2hdmin3}
\mathop{\min}_{  t_{k32}\geq 0}\quad  f(t_{k32})\triangleq t_2+ \left( 2^{\frac{D_{\pi_k}}{B t_{k32}}} -1 \right) \frac{ t_{k32}}
{u_2}.
\end{equation}
Then,
\begin{equation}\label{sf2hdmin3_1}
 f''(t_{k32}) = \frac{(\ln 2)^2 }
{ u_2 B^2 t_{k32}^3}\text e^{\frac{(\ln 2)D_{\pi_k}}{B t_{k32}}}\geq 0,
\end{equation}
which verifies that problem \eqref{sf2hdmin3} is convex,
and the optimal $t_{k32}$ can be obtained by solving $f'(t_{k32})=0$.
Calculating from \eqref{sf2hdmin3}, we have
\begin{equation}\label{sf2hdeq3}
f'(t_{k32})= 1 \!-\! \left(\!\left(\!\frac{(\ln 2)D_{\pi_k}}{Bt_{k32}} -1\!\right)\text e^{\frac{(\ln 2)D_{\pi_k}}{B t_{k32}}} +1\! \right)\frac{1}
{u_2},
\end{equation}
which yields
\begin{equation}\label{sf2hdeq5}
t_{k32}^*= \frac{ (\ln 2)D_{\pi_k}}{ B W \left(  \frac{u_2-1}
{\text e  }\right) +B}.
\end{equation}

As a result, the optimal solution of problem \eqref{ee2min2_2} is provided in Theorem 2.


\section{Proof of Theorem 3}
\setcounter{equation}{0}
\renewcommand{\theequation}{\thesection.\arabic{equation}}
To minimize problem  \eqref{ee2min3_2}, constraints (\ref{ee2min3_2}b) and (\ref{ee2min3_2}c) should hold with equalities for the optimal solution, as otherwise the objective value (\ref{ee2min3_2}a) can be further decreased, which contradicts that the solution is optimal.
Setting (\ref{ee2min3_2}b) and (\ref{ee2min3_2}c) with equalities yields
\begin{equation}\label{sf2fdeq1}
p_{\pi_k}=\left( 2^{\frac{D}{B (t_{k3} -\delta_{\pi_k})}} -1 \right) \frac{\gamma P +B\sigma^2}{b  h_k^{-\alpha}}.
\end{equation}
\begin{equation}\label{sf2fdeq2}
p_{\pi_k}=\frac{\epsilon_{\pi_k}( \zeta  P b  h_k^{-\alpha} -p_{\pi_k}^{\text{re}} )t_{k3}}{t_{k3}-\delta_{\pi_k}} -\epsilon_{\pi_k}p_{\pi_k}^{\text{tr}}.
\end{equation}

Combining \eqref{sf2fdeq1} and \eqref{sf2fdeq2}, the optimal solution of problem  \eqref{ee2min3_2} is given by \eqref{sf2fdeq3_1} and \eqref{sf2fdeq3_2}.

%
%
%
%
%
%

\bibliographystyle{IEEEtran}
\bibliography{IEEEabrv,MMM}

\begin{thebibliography}{10}
\providecommand{\url}[1]{#1}
\csname url@samestyle\endcsname
\providecommand{\newblock}{\relax}
\providecommand{\bibinfo}[2]{#2}
\providecommand{\BIBentrySTDinterwordspacing}{\spaceskip=0pt\relax}
\providecommand{\BIBentryALTinterwordstretchfactor}{4}
\providecommand{\BIBentryALTinterwordspacing}{\spaceskip=\fontdimen2\font plus
\BIBentryALTinterwordstretchfactor\fontdimen3\font minus
  \fontdimen4\font\relax}
\providecommand{\BIBforeignlanguage}[2]{{%
\expandafter\ifx\csname l@#1\endcsname\relax
\typeout{** WARNING: IEEEtran.bst: No hyphenation pattern has been}%
\typeout{** loaded for the language `#1'. Using the pattern for}%
\typeout{** the default language instead.}%
\else
\language=\csname l@#1\endcsname
\fi
#2}}
\providecommand{\BIBdecl}{\relax}
\BIBdecl

\bibitem{7470933}
Y.~Zeng, R.~Zhang, and T.~J. Lim, ``Wireless communications with unmanned
  aerial vehicles: {O}pportunities and challenges,'' \emph{IEEE Commun. Mag.},
  vol.~54, no.~5, pp. 36--42, May 2016.

\bibitem{DBLP:journals/corr/abs-1805-06532}
\BIBentryALTinterwordspacing
M.~Mozaffari, A.~T.~Z. Kasgari, W.~Saad, M.~Bennis, and M.~Debbah, ``Beyond {5G
  with UAVs: Foundations of a 3D} wireless cellular network,'' \emph{CoRR},
  vol. abs/1805.06532, 2018. [Online]. Available:
  \url{http://arxiv.org/abs/1805.06532}
\BIBentrySTDinterwordspacing

\bibitem{saad2019vision}
W.~Saad, M.~Bennis, and M.~Chen, ``A vision of {6G} wireless systems:
  {A}pplications, trends, technologies, and open research problems,''
  \emph{arXiv preprint arXiv:1902.10265}, 2019.

\bibitem{7936620}
R.~Amorim, H.~Nguyen, P.~Mogensen, I.~Z. Kov\'{a}cs, J.~Wigard, and T.~B.
  S{\o}rensen, ``Radio channel modeling for {UAV} communication over cellular
  networks,'' \emph{IEEE Wireless Commun. Lett.}, vol.~6, no.~4, pp. 514--517,
  Aug. 2017.

\bibitem{8048502}
A.~Al-Hourani and K.~Gomez, ``Modeling cellular-to-{UAV} path-loss for suburban
  environments,'' \emph{IEEE Wireless Commun. Lett.}, vol.~7, no.~1, pp.
  82--85, Feb. 2018.

\bibitem{8247211}
Q.~Wu, Y.~Zeng, and R.~Zhang, ``Joint trajectory and communication design for
  multi-{UAV} enabled wireless networks,'' \emph{IEEE Trans. Wireless Commun.},
  vol.~17, no.~3, pp. 2109--2121, Mar. 2018.

\bibitem{8353131}
X.~Wang, K.~Wang, S.~Wu, D.~Sheng, H.~Jin, K.~Yang, and S.~Ou, ``Dynamic
  resource scheduling in mobile edge cloud with cloud radio access network,''
  \emph{IEEE Trans. Parallel Distributed Syst.}, pp. 1--1, 2018.

\bibitem{8379427}
Z.~Yang, C.~Pan, M.~Shikh-Bahaei, W.~Xu, M.~Chen, M.~Elkashlan, and
  A.~Nallanathan, ``Joint altitude, beamwidth, location, and bandwidth
  optimization for {UAV}-enabled communications,'' \emph{IEEE Commun. Lett.},
  vol.~22, no.~8, pp. 1716--1719, Aug. 2018.

\bibitem{7469804}
S.~Say, H.~Inata, J.~Liu, and S.~Shimamoto, ``Priority-based data gathering
  framework in {UAV}-assisted wireless sensor networks,'' \emph{IEEE Sensors
  J.}, vol.~16, no.~14, pp. 5785--5794, July 2016.

\bibitem{8119562}
C.~Zhan, Y.~Zeng, and R.~Zhang, ``Energy-efficient data collection in {UAV}
  enabled wireless sensor network,'' \emph{IEEE Wireless Commun. Lett.},
  vol.~7, no.~3, pp. 328--331, June 2018.

\bibitem{8432487}
J.~Gong, T.~Chang, C.~Shen, and X.~Chen, ``Flight time minimization of {UAV}
  for data collection over wireless sensor networks,'' \emph{IEEE J. Sel. Areas
  Commun.}, pp. 1--1, 2018.

\bibitem{8316986}
D.~Yang, Q.~Wu, Y.~Zeng, and R.~Zhang, ``Energy tradeoff in ground-to-{UAV}
  communication via trajectory design,'' \emph{IEEE Trans. Veh. Technol.},
  vol.~67, no.~7, pp. 6721--6726, July 2018.

\bibitem{8489918}
L.~{Xie}, J.~{Xu}, and R.~{Zhang}, ``Throughput maximization for {UAV}-enabled
  wireless powered communication networks,'' \emph{IEEE Internet Things J.},
  vol.~6, no.~2, pp. 1690--1703, Apr. 2019.

\bibitem{8417659}
------, ``Throughput maximization for {UAV}-enabled wireless powered
  communication networks - invited paper,'' in \emph{Proc. IEEE Veh. Technol.
  Conf.}, June 2018, pp. 1--7.

\bibitem{8540379}
S.~{Cho}, K.~{Lee}, B.~{Kang}, K.~{Koo}, and I.~{Joe}, ``Weighted
  harvest-then-transmit: {UAV}-enabled wireless powered communication
  networks,'' \emph{IEEE Access}, vol.~6, pp. 72\,212--72\,224, 2018.

\bibitem{8736248}
F.~{Wu}, D.~{Yang}, L.~{Xiao}, and L.~{Cuthbert}, ``Energy consumption and
  completion time tradeoff in rotary-wing {UAV} enabled {WPCN},'' \emph{IEEE
  Access}, pp. 1--1, 2019.

\bibitem{wu2019minimum}
F.~Wu, D.~Yang, L.~Xiao, and L.~Cuthbert, ``Minimum-throughput maximization for
  multi-{UAV}-enabled wireless-powered communication networks,''
  \emph{Sensors}, vol.~19, no.~7, p. 1491, 2019.

\bibitem{8365881}
J.~Xu, Y.~Zeng, and R.~Zhang, ``{UAV}-enabled wireless power transfer:
  Trajectory design and energy optimization,'' \emph{IEEE Trans. Wireless
  Commun.}, pp. 1--1, 2018.

\bibitem{hu2018optimal}
Y.~Hu, X.~Yuan, J.~Xu, and A.~Schmeink, ``Optimal 1{D} trajectory design for
  {UAV}-enabled multiuser wireless power transfer,'' \emph{arXiv preprint
  arXiv:1811.00471}, 2018.

\bibitem{8434285}
F.~Zhou, Y.~Wu, R.~Q. Hu, and Y.~Qian, ``Computation rate maximization in
  uav-enabled wireless powered mobile-edge computing systems,'' \emph{IEEE J.
  Sel. Areas Commun.}, pp. 1--1, 2018.

\bibitem{5937283}
P.~Zhan, K.~Yu, and A.~L. Swindlehurst, ``Wireless relay communications with
  unmanned aerial vehicles: {P}erformance and optimization,'' \emph{IEEE Trans.
  Aerosp. Electron. Syst.}, vol.~47, no.~3, pp. 2068--2085, July 2011.

\bibitem{7959158}
L.~Kong, L.~Ye, F.~Wu, M.~Tao, G.~Chen, and A.~V. Vasilakos, ``Autonomous relay
  for millimeter-wave wireless communications,'' \emph{IEEE J. Sel. Areas
  Commun.}, vol.~35, no.~9, pp. 2127--2136, Sept. 2017.

\bibitem{8278204}
R.~Fan, J.~Cui, S.~Jin, K.~Yang, and J.~An, ``Optimal node placement and
  resource allocation for {UAV} relaying network,'' \emph{IEEE Commun. Lett.},
  vol.~22, no.~4, pp. 808--811, Apr. 2018.

\bibitem{8629002}
C.~{Pan}, H.~{Ren}, Y.~{Deng}, M.~{Elkashlan}, and A.~{Nallanathan}, ``Joint
  blocklength and location optimization for {URLLC}-enabled {UAV} relay
  systems,'' \emph{IEEE Commun. Lett.}, vol.~23, no.~3, pp. 498--501, Mar.
  2019.

\bibitem{7412759}
M.~Mozaffari, W.~Saad, M.~Bennis, and M.~Debbah, ``Unmanned aerial vehicle with
  underlaid device-to-device communications: {P}erformance and tradeoffs,''
  \emph{IEEE Trans. Wireless Commun.}, vol.~15, no.~6, pp. 3949--3963, June
  2016.

\bibitem{8254370}
N.~Zhao, F.~Cheng, F.~R. Yu, J.~Tang, Y.~Chen, G.~Gui, and H.~Sari, ``Caching
  {UAV} assisted secure transmission in hyper-dense networks based on
  interference alignment,'' \emph{IEEE Trans. Commun.}, vol.~66, no.~5, pp.
  2281--2294, May 2018.

\bibitem{7875131}
M.~Chen, M.~Mozaffari, W.~Saad, C.~Yin, M.~Debbah, and C.~S. Hong, ``Caching in
  the sky: {P}roactive deployment of cache-enabled unmanned aerial vehicles for
  optimized quality-of-experience,'' \emph{IEEE J. Sel. Areas Commun.},
  vol.~35, no.~5, pp. 1046--1061, May 2017.

\bibitem{8764580}
Z.~{Yang}, C.~{Pan}, K.~{Wang}, and M.~{Shikh-Bahaei}, ``Energy efficient
  resource allocation in {UAV}-enabled mobile edge computing networks,''
  \emph{IEEE Trans. Wireless Commun.}, vol.~18, no.~9, pp. 4576--4589, Sep.
  2019.

\bibitem{6863654}
A.~Al-Hourani, S.~Kandeepan, and S.~Lardner, ``Optimal {LAP} altitude for
  maximum coverage,'' \emph{IEEE Wireless Commun. Lett.}, vol.~3, no.~6, pp.
  569--572, Dec. 2014.

\bibitem{7918510}
M.~Alzenad, A.~El-Keyi, F.~Lagum, and H.~Yanikomeroglu, ``3-{D} placement of an
  unmanned aerial vehicle base station ({UAV-BS}) for energy-efficient maximal
  coverage,'' \emph{IEEE Wireless Commun. Lett.}, vol.~6, no.~4, pp. 434--437,
  Aug. 2017.

\bibitem{8038014}
M.~Alzenad, A.~El-Keyi, and H.~Yanikomeroglu, ``3-{D} placement of an unmanned
  aerial vehicle base station for maximum coverage of users with different
  {QoS} requirements,'' \emph{IEEE Wireless Commun. Lett.}, vol.~7, no.~1, pp.
  38--41, Feb. 2018.

\bibitem{7762053}
J.~Lyu, Y.~Zeng, R.~Zhang, and T.~J. Lim, ``Placement optimization of
  {UAV}-mounted mobile base stations,'' \emph{IEEE Commun. Lett.}, vol.~21,
  no.~3, pp. 604--607, Mar. 2017.

\bibitem{7510820}
R.~I. Bor-Yaliniz, A.~El-Keyi, and H.~Yanikomeroglu, ``Efficient 3-{D}
  placement of an aerial base station in next generation cellular networks,''
  in \emph{Proc. IEEE Int. Conf. Commun.}, May 2016, pp. 1--5.

\bibitem{7486987}
M.~Mozaffari, W.~Saad, M.~Bennis, and M.~Debbah, ``Efficient deployment of
  multiple unmanned aerial vehicles for optimal wireless coverage,'' \emph{IEEE
  Commun. Lett.}, vol.~20, no.~8, pp. 1647--1650, Aug. 2016.

\bibitem{7569080}
V.~{Sharma}, R.~{Sabatini}, and S.~{Ramasamy}, ``{UAV}s assisted delay
  optimization in heterogeneous wireless networks,'' \emph{IEEE Commun. Lett.},
  vol.~20, no.~12, pp. 2526--2529, Dec. 2016.

\bibitem{8438896}
Q.~{Wu} and R.~{Zhang}, ``Common throughput maximization in {UAV}-enabled
  {OFDMA} systems with delay consideration,'' \emph{IEEE Trans. Commun.},
  vol.~66, no.~12, pp. 6614--6627, Dec. 2018.

\bibitem{8779596}
C.~{Zhan} and Y.~{Zeng}, ``Completion time minimization for multi-{UAV}-enabled
  data collection,'' \emph{IEEE Trans. Wireless Commun.}, vol.~18, no.~10, pp.
  4859--4872, Oct. 2019.

\bibitem{8419316}
C.~{Zhan}, Y.~{Zeng}, and R.~{Zhang}, ``Trajectory design for distributed
  estimation in {UAV}-enabled wireless sensor network,'' \emph{IEEE Trans. Veh.
  Technol.}, vol.~67, no.~10, pp. 10\,155--10\,159, Oct. 2018.

\bibitem{7888557}
Y.~Zeng and R.~Zhang, ``Energy-efficient {UAV} communication with trajectory
  optimization,'' \emph{IEEE Trans. Wireless Commun.}, vol.~16, no.~6, pp.
  3747--3760, June 2017.

\bibitem{6907966}
H.~Ju and R.~Zhang, ``Optimal resource allocation in full-duplex
  wireless-powered communication network,'' \emph{IEEE Trans. Commun.},
  vol.~62, no.~10, pp. 3528--3540, Oct. 2014.

\bibitem{7081080}
Z.~Ding, C.~Zhong, D.~W.~K. Ng, M.~Peng, H.~A. Suraweera, R.~Schober, and H.~V.
  Poor, ``Application of smart antenna technologies in simultaneous wireless
  information and power transfer,'' \emph{IEEE Commun. Mag.}, vol.~53, no.~4,
  pp. 86--93, Apr. 2015.

\bibitem{6623062}
X.~Zhou, R.~Zhang, and C.~K. Ho, ``Wireless information and power transfer:
  {A}rchitecture design and rate-energy tradeoff,'' \emph{IEEE Trans. Commun.},
  vol.~61, no.~11, pp. 4754--4767, Nov. 2013.

\bibitem{DBLP:journals/corr/abs-1803-07123}
B.~{Clerckx}, R.~{Zhang}, R.~{Schober}, D.~W.~K. {Ng}, D.~I. {Kim}, and H.~V.
  {Poor}, ``Fundamentals of wireless information and power transfer: From rf
  energy harvester models to signal and system designs,'' \emph{IEEE J. Sel.
  Areas Commun.}, vol.~37, no.~1, pp. 4--33, Jan. 2019.

\bibitem{8241822}
J.~Kang, I.~Kim, and D.~I. Kim, ``Wireless information and power transfer:
  {R}ate-energy tradeoff for nonlinear energy harvesting,'' \emph{IEEE Trans.
  Wireless Commun.}, vol.~17, no.~3, pp. 1966--1981, Mar. 2018.

\bibitem{8294215}
Z.~{Yang}, W.~{Xu}, Y.~{Pan}, C.~{Pan}, and M.~{Chen}, ``Optimal fairness-aware
  time and power allocation in wireless powered communication networks,''
  \emph{IEEE Trans. Commun.}, vol.~66, no.~7, pp. 3122--3135, July 2018.

\bibitem{zeng2019energy}
Y.~Zeng, J.~Xu, and R.~Zhang, ``Energy minimization for wireless communication
  with rotary-wing uav,'' \emph{IEEE Trans. Wireless Commun.}, 2019.

\bibitem{boyd2004convex}
S.~Boyd and L.~Vandenberghe, \emph{Convex {O}ptimization}.\hskip 1em plus 0.5em
  minus 0.4em\relax Cambridge University Press, 2004.

\bibitem{8088359}
Z.~{Yang}, W.~{Xu}, J.~{Shi}, H.~{Xu}, and M.~{Chen}, ``Association and load
  optimization with user priorities in load-coupled heterogeneous networks,''
  \emph{IEEE Trans. Wireless Commun.}, vol.~17, no.~1, pp. 324--338, Jan. 2018.

\bibitem{bertsekas2009convex}
D.~P. Bertsekas, \emph{{Convex Optimization Theory}}.\hskip 1em plus 0.5em
  minus 0.4em\relax Athena Scientific Belmont, 2009.

\bibitem{Yang2019OptimizationOR}
Z.~Yang, M.~Chen, W.~Saad, and M.~Shikh-Bahaei, ``Optimization of rate
  allocation and power control for rate splitting multiple access ({RSMA}),''
  \emph{arXiv preprint arXiv:1903.08068}, 2019.

\bibitem{filippone2006flight}
A.~Filippone, \emph{Flight performance of fixed and rotary wing
  aircraft}.\hskip 1em plus 0.5em minus 0.4em\relax Elsevier, 2006.

\bibitem{bramwell2001bramwell}
A.~R.~S. Bramwell, D.~Balmford, and G.~Done, \emph{Bramwell's helicopter
  dynamics}.\hskip 1em plus 0.5em minus 0.4em\relax Elsevier, 2001.

\end{thebibliography}

\end{document}